\PassOptionsToPackage{noadjust}{cite}
\documentclass[11pt,a4paper]{article}
\usepackage[margin=2.5cm]{geometry}
\usepackage{bm}
\usepackage{chngcntr}
\usepackage{amsfonts, amsmath, amsthm, amssymb}
\usepackage[all,cmtip]{xy}
\usepackage{appendix}
\usepackage{array}
\newcolumntype{C}[1]{>{\centering\let\newline\\\arraybackslash$}m{#1}<{$\centering}} 
\DeclareMathAlphabet{\mathscr}{OT1}{pzc}{m}{it}
\newcommand{\C}{\mathbb C}
\newcommand{\R}{\mathbb R}
\newcommand{\Z}{\mathbb Z}
\newcommand{\f}{\mathfrak}
\newcommand{\ca}{\mathscr}

\renewcommand{\atop}[2]{{\genfrac{}{}{0pt}{}{#1}{#2}}}
\newcommand{\newfn}{\DeclareMathOperator}
\newfn{\Tr}{Tr}
\newfn{\End}{End}
\newfn{\Hom}{Hom}
\newfn{\Sp}{Sp}
\newfn{\sgn}{sgn}
\newfn{\ii}{i}
\newfn{\e}{e}
\newfn{\dcm}{d}
\newfn{\n}{n}
\renewcommand{\L}{\mathrm L}
\newcommand{\dd}{\, \mathrm d}

\newcommand{\Vac}{\psi_{\emptyset}}
\newcommand{\set}[2]{\left\{ \, {#1} \, \middle| \, {#2}  \, \right\} }
\newcommand{\inner}[2]{\left\langle {#1}, {#2} \right \rangle }
\newcommand{\innerrnd}[2]{\left({#1}, {#2}\right) }
\newtheorem{thm}{Theorem}[section]
\newtheorem{cor}[thm]{Corollary}
\newtheorem{lem}[thm]{Lemma}
\newtheorem{prop}[thm]{Proposition}
\newtheorem{rem}[thm]{Remark}
\newtheorem{defn}[thm]{Definition}

\newcommand{\rft}[1]{Thm. \ref{#1}}
\newcommand{\rfl}[1]{Lemma \ref{#1}}
\newcommand{\rfp}[1]{Prop. \ref{#1}}
\def\cprime{$'$}

\title{A non-symmetric Yang-Baxter Algebra for the Quantum Nonlinear Schr\"odinger Model}
\author{Bart Vlaar\thanks{E-mail: B.H.M.Vlaar@uva.nl} \\ Korteweg de Vries Institute for Mathematics \\ University of Amsterdam, Netherlands}
\date{\today}

\begin{document}

\maketitle

\begin{abstract}
We study certain non-symmetric wavefunctions associated to the quantum nonlinear Schr\"odinger model, introduced by Komori and Hikami using Gutkin's propagation operator, which involves representations of the degenerate affine Hecke algebra. 
We highlight how these functions can be generated using a vertex-type operator formalism similar to the recursion defining the symmetric (Bethe) wavefunction in the quantum inverse scattering method. 
Furthermore, some of the commutation relations encoded in the Yang-Baxter equation for the relevant monodromy matrix are generalized to the non-symmetric case.
\end{abstract}

\section{Introduction}

The quantum nonlinear Schr\"odinger (QNLS) or Lieb-Liniger model was introduced in 1963 \cite{LiebLi} and has been studied extensively since, e.g \cite{Dorlas, Emsiz, EmsizOS, Gaudin1971-1, Gaudin1971-2, Gaudin1971-3, Gaudin1983, Gutkin1982, Gutkin1985, Gutkin1988, HeckmanOpdam1997, Hikami, KomoriHikami, Korepin, Sklyanin1989, YangYang}. 
It describes a system of $N$ spinless (in particular, bosonic) particles restricted to a circle or an infinite line with pairwise contact interaction whose strength is determined by a constant $\gamma \in \R$; most of the theory deals with the repulsive case ($\gamma>0$). 
%Nevertheless the QNLS model still has open questions attached to it, and dealing with these issues is all the more important since 
In many ways the QNLS model is a prototypical integrable model; it was introduced \cite{LiebLi} as the first example of a parameter-dependent boson gas for which eigenstates and eigenvalues of the quantum Hamiltonian can be calculated exactly. 
Earlier, Girardeau \cite{Girardeau} studied a related system without a (nontrivial) parameter but which can be obtained from the QNLS model in the limit $\gamma \to \infty$. 
Furthermore, there has been experimental interest; the low energy eigenstates of a certain gas of three-dimensional particles in a long cylinder are described by the QNLS model \cite{LiebSY,SeiringerY} and such systems have been manufactured \cite{Amerongen,AmerongenEWKD} by magnetically trapping and cooling rubidium-85 atoms. \\

Consider the standard Euclidean basis of $\R^N$ consisting of the vectors $\bm e_1,\ldots,\bm e_N$. 
Assume the particle coordinates are given by $\bm x = (x_1,\ldots,x_N) \in J^N$ where $J=[x^+,x^-] \subset \R$ and write $\partial_j = \frac{\partial}{\partial x_j}$.
In convenient units, the Hamiltonian for the QNLS model is formally given by 
\[ H^\gamma_{(N)} = - \sum_{j=1}^N \partial_j^2 + 2\gamma \sum_{1 \leq j < k \leq N} \delta(x_j-x_k). \]
Its eigenfunctions will be referred to as \emph{wavefunctions}; the eigenvalue problem for $H^\gamma_{(N)}$ can be made rigorous \cite{LiebLi} by replacing it with a Helmholtz equation and imposing certain jump conditions on the derivatives of the candidate functions. \\

We note that a physically acceptable wavefunction $\Psi$ describing a bosonic system must be symmetric: $\Psi(x_1,\ldots,x_N) = \Psi(x_{w1},\ldots,x_{wN})$ for all $w \in S_N$. 
Notwithstanding this, we will consider the non-symmetric\footnote{Where suitable we will denote non-symmetric objects by lowercase letters and their symmetric counterparts by the corresponding capital letters.} eigenfunctions $\psi$ of the QNLS Hamiltonian, brought into the theoretical picture of the QNLS model by Komori and Hikami \cite{Hikami,KomoriHikami} by means of Gutkin's \emph{propagation operator} \cite{EmsizOS,Gutkin1982}, which intertwines two representations of the A-type \emph{degenerate affine Hecke algebra} (dAHA) \cite{EmsizOS,Gutkin1982,HeckmanOpdam1997}, in analogy to the non-symmetric Jack polynomials in the Calogero-Sutherland-Moser model \cite{BakerForrester1997,BernardGHP,Hikami1996,Polychronakos}. 
From the non-symmetric wavefunctions the symmetric ones are obtained by symmetrization: 
$\Psi(\bm x)  = \frac{1}{N!} \sum_{w \in S_N} \psi(x_{w 1},\ldots,x_{w N})$. 
Our main result is a recursive formula generating the non-symmetric wavefunctions with vertex-type operators, akin to the formula recursively defining the \emph{symmetric} (Bethe) wavefunction in the quantum inverse scattering method \cite{Gutkin1988,KorepinBI,Sklyanin1982}, thus closely tying this method to the Hecke algebra approach, and emphasizing the importance of the non-symmetric wavefunctions in the theoretical understanding of the QNLS model. 
Similarly, a recursive structure is known for the non-symmetric Jack polynomials featuring in the Calogero-Sutherland-Moser model \cite{BakerForrester1997,KnopSahi}.

\subsection{Outline}

We will recall how to treat the QNLS Hamiltonian eigenvalue problem more rigorously in Sect. \ref{sec:2}. We will also discuss the history of the solution methods of the QNLS model, in particular we will review Lieb and Liniger's solution and briefly discuss the quantum inverse scattering method (QISM) for the QNLS model.
In Sect. \ref{sec:3} of the present paper we will discuss aspects of the dAHA approach for the pertinent case (A$_{N\!-\!1}$-type) in more detail. 
In particular, we will review the \emph{propagation operator} and the \emph{non-symmetric eigenfunctions} $\psi_{\bm \lambda}$ alluded to earlier. 
In Sect. \ref{sec:4} we will define two \emph{non-symmetric creation operators} $b^\pm_\mu$ that can be used 
to generate the $\psi_{\bm \lambda}$ recursively, by virtue of convenient commutation relations with the propagation operator.
In Sect. \ref{sec:5} for the case of the QNLS problem on the circle we will define operators $a^\pm_\mu,c^\pm_\mu$ that together with $b^\pm_\mu$ satisfy certain commutation relations.
The connection between the ``non-symmetric'' operators $a^\pm_\mu,b^\pm_\mu,c^\pm_\mu$ and their established ``symmetric'' equivalents, the operators $A_\mu$, $B_\mu$, $C_\mu$, $D_\mu$ from the QISM for the QNLS model is made in Sect. \ref{sec:6}. 
Some well-known commutation relations of these symmetric operators are recovered. 
We conclude in Sect. \ref{sec:7} by summarizing the main results, unresolved issues and possible applications and generalizations.

\subsection*{Acknowledgements}

This paper is based on the author's PhD thesis \cite{VlaarThesis} supervised by Christian Korff at the University of Glasgow, undertaken with financial support from the Engineering and Physical Sciences Research Council.
This work was partially supported by a free competition grant of the Netherlands Organization for Scientific Research (NWO).
The author would like to thank C. Korff, J. Stokman and N. Reshetikhin for useful discussions and suggestions. 
The author is grateful to the referees for their comments.

\section{Rigorous definition of the QNLS model and solution methods } \label{sec:2}

Throughout this paper we will let $\gamma \in \R$, $N \in \Z_{\geq 0}$ and $J = [x^+,x^-] \subset \R$ be fixed but arbitrary. 
The interval $J$ may be unbounded; if it is bounded we will denote $L=x^--x^+>0$.

\subsection{Hyperplanes and derivative jump conditions}

To place $H^\gamma_{(N)}$ on a more rigorous footing, consider the standard A$_{N\!-\!1}$ hyperplane arrangement 
\[ \left\{ V_{j \,k} := (\bm e_j-\bm e_k)^\perp = \set{\bm x \in \R^N}{x_j = x_k} \right\}_{1 \leq j < k \leq N} \]
and\footnote{The vectors $\bm e_j - \bm e_k$, $1 \! \leq \! j \! < \! k \! \leq \! N$, realize a positive system of a finite root system of A$_{N-1}$-type, spanning the subset of $\R^N$ of vectors whose coordinates sum to zero (corresponding to studying the $N$-particle system in the centre-of-mass frame). It is a peculiarity of A-type root systems that their natural realizations do not span the whole coordinate space. In this paper, we will work with Weyl group actions on the whole $\R^N$. }
the associated set of \emph{regular vectors}
\[ \R^N_\text{reg} := \R^N \setminus \bigcup_{1 \leq j ,k \leq N} V_{j \, k} = \set{\bm x \in \R^N}{x_j \ne x_k \text{ if } j \ne k}. \] 
$S_N$ is the Weyl group associated to the collection of hyperplanes $V_{j \, k}$ in the following way. 
Given $1 \leq j \ne k \leq N$, write $s_{j \, k}$ for the transposition swapping $j$ and $k$, and for $j=1,\ldots,N-1$
write $s_j := s_{j \, j+1}$.
Then the orthogonal reflection in the hyperplane $V_{j \, k}$ is given by
\[ s_{j \, k} (x_1,\ldots, x_N) = (x_1,\ldots, \underset{(j)}{x_k}, \ldots, \underset{(k)}{x_j}, \ldots, x_N), \]
which can be extended to a left action of $S_N$ on $\R^N$.
$\R^N_\text{reg}$ is an invariant subset and this action carries over to the set of its connected components, the \emph{alcoves}.
In fact, we have
\begin{equation} \R^N_\text{reg} = \bigcup_{w \in S_N} w \R^N_+, \qquad \R^N_+ = \set{\bm x \in \R^N}{x_1 > \ldots > x_N}. \label{eqn:fundamentalalcove} \end{equation}

Let $U \subset \C^N$. We use the standard notations $\ca F(U),\ca C(U),\ca C^r(U)$ for the vector spaces of functions, continuous functions and, for $r \in \Z_{\geq 1} \cup \{ \infty \}$ and open $U$, $r$-times continuously differentiable functions: $U \to \C$, respectively. If $U$ is an $S_N$-invariant subset of $\C^N$ such as $\R^N$ or $\R^N_\text{reg}$ these are left $S_N$-modules through $(w f)(\bm x) = f(w^{-1} \bm x)$.
Furthermore, the following $S_N$-submodules of $\ca C(\R^N)$ have been introduced \cite{EmsizOS,Gutkin1982}:
\begin{align*}
\ca{CB}^1(\R^N) &= \left\{ \, f \in \ca C(\R^N) \, \middle| \, \forall w \, f|_{w \R^N_+} \text{ has a } \ca{C}^1\text{-extension to an open neighbourhood of } \overline{w \R^N_+} \, \right\},  \\ 
\ca{CB}^\infty(\R^N) &= \set{f \in \ca C(\R^N)}{\forall w \, f|_{w \R^N_+} \text{ is smooth}}.
\end{align*}
Given the above action of $S_N$ on invariant subsets $X \subset \ca C(\R^N)$, the set of $S_N$-invariant (i.e. symmetric) elements of $X$ is denoted $X^{S_N}$. 
The action of the \emph{group} $S_N$ can be linearly extended to an action of the \emph{group algebra} $\C S_N$, which is the algebra consisting of formal linear combinations $\sum_{w \in S_N} c_w w$ where each $c_w \in \C$. 
A particularly important element of $\C S_N$ for the present purposes is the \emph{symmetrizer} $\ca S_{(N)} := \frac{1}{N!} \sum_{w \in S_N} w$, which sends elements of invariant subsets $X \subset \ca C(\R^N)$ to elements in $X^{S_N}$. \\

Given $1 \leq j < k \leq N$, call $\bm x \in V_{j \, k}$ \emph{subregular} if $\bm x$ is not on any other hyperplane. 
For subregular $\bm x \in V_{j \, k}$ and $\delta>0$ small enough, $\bm x_{\pm \delta} := \bm x \pm \delta (\bm e_j-\bm e_k)$ is in an alcove, $w^{-1} \R^N$ say, the walls of which are subsets with nonempty interior of the hyperplanes $V_{w^{-1}(r) \, w^{-1}(r+1)}$ where $r=1,\ldots,N-1$. Hence $w(j)=r$ and $w(k)=r+1$ for some $r=1,\ldots,N-1$. We now recall the following key result.
\begin{prop}\cite[Prop. 2.2]{EmsizOS} 
Let $f \in \ca{CB}^1(\R^N)$ and $\gamma,E\in \R$. 
Then $f$ is an eigenfunction of $H^\gamma_{(N)}$ with eigenvalue $E$ precisely if $f$ satisfies
\begin{equation} - \sum_{j=1}^N \partial_j^2 f|_{\R^N_\text{reg}} = E f|_{\R^N_\text{reg}} \label{eqn:QNLSLaplace} \end{equation}  
and the \emph{derivative jump conditions}
\begin{equation}\label{eqn:QNLSjump}  
\lim_{\atop{\delta \to 0}{\delta>0}} \Bigl( \left(\partial_j-\partial_k\right) f(\bm x_\delta) -\left(\partial_j-\partial_k\right) f(\bm x_{-\delta}) \Bigr) =  2\gamma f(\bm x), \end{equation}
for $1 \leq j < k \leq N$ and $\bm x \in V_{j \, k}$ subregular.
A priori, \eqref{eqn:QNLSLaplace} is interpreted distributionally; however, if $f$ also satisfies \eqref{eqn:QNLSjump} then $f \in \ca{CB}^\infty(\R^N)$.
\end{prop}

For $f$ to describe a system of particles on a circle of circumference $L$, say, then the finite hyperplane arrangement $\{ V_{j \, k} \}_{1 \leq j < k \leq N}$ needs to be replaced by the affine hyperplane arrangement 
\[ \left\{ V_{j \, k; m} := \set{\bm x \in \R^N}{x_j-x_k = m L} \right\}_{1 \leq j < k \leq N, \, m \in \Z}. \]
A key role is played by the affine Weyl group $\hat S_N = \langle s_0,\ldots,s_N \rangle$, where the affine simple reflection $s_0$ acts as $s_0 (\bm x) = (x_N+L,x_2,\ldots,x_{N-1},x_1-L)$.
Furthermore, $f$ is required to be $L$-periodic in each variable. We refer to \cite{Emsiz,EmsizOS} for more detail.
We will follow an alternative approach \cite{LiebLi}, where we take a solution $f$ of the QNLS problem on $\R^N$, i.e. satisfying \eqref{eqn:QNLSLaplace}-\eqref{eqn:QNLSjump}, consider its restriction to a hypercube $J^N$, where $J=[x^+,x^-]$ with $x^--x^+=L$, and impose 
\begin{equation} \label{eqn:QNLSperiodic}
\begin{aligned}
f(\bm x)|_{x_j=x^+} &= f(\bm x)|_{x_j=x^-}, \\
\lim_{x_j \to x^+} \partial_j f(\bm x) &= \lim_{x_j \to x^-} \partial_j f(\bm x),
\end{aligned}
\qquad \text{for } j =1,\ldots,N.
\end{equation}

\subsection{The Bethe ansatz}

Lieb and Liniger \cite{LiebLi} solved the QNLS problem (both on the line and the circle) by modifying Bethe's approach for analysing the one-dimensional Heisenberg model \cite{Bethe}, now known as the \emph{(coordinate) Bethe ansatz} (BA). 
Write $\ii = \sqrt{-1}$ and $\inner{\bm w}{\bm z} = \sum_{j=1}^N w_j \bar z_j$ for the Euclidean complex inner product on $\C^N$.
Consider the \emph{plane wave} $\e^{\ii \bm \lambda} \in \ca C^\infty(\R^N)$ with wavevector $\bm \lambda = (\lambda_1,\ldots,\lambda_N) \in \C^N$ defined by $\e^{\ii \bm \lambda}(\bm x) = \e^{\ii \inner{\bm \lambda}{\bm x}}$.
The BA results in
\begin{prop}\cite{LiebLi} \label{prop:Bethewavefn1}
The function $\Psi_{\bm \lambda} \in \ca{CB}^\infty(\R^N)^{S_N}$ defined by
\begin{equation} \label{eqn:Bethewavefn1} \Psi_{\bm \lambda}|_{\R^N_+} = \frac{1}{N!} \sum_{w \in S_N} G^\gamma_{w \bm \lambda} \e^{\ii (w \bm \lambda)}
\end{equation}
satisfies \eqref{eqn:QNLSLaplace}-\eqref{eqn:QNLSjump} with $E=\sum_j \lambda_j^2$ precisely if $G^\gamma_{\bm \lambda} = \prod_{j<k} \frac{\lambda_{j}-\lambda_{k}-\ii \gamma}{\lambda_{j}-\lambda_{k}}$.
If in addition the $\lambda_j$ are distinct and satisfy the \emph{Bethe ansatz equations} (BAEs), viz.
\begin{equation} \label{eqn:BAEintro} \e^{\ii \lambda_j L} = \prod_{\atop{k=1}{k \ne j}}^N \frac{\lambda_j-\lambda_k + \ii \gamma}{\lambda_j-\lambda_k - \ii \gamma}, \quad \text{for } j =1,\ldots, N, 
\end{equation}
then $\Psi_{\bm \lambda}|_{J^N}$ satisfies \eqref{eqn:QNLSperiodic}.
\end{prop}

\subsection{The quantum inverse scattering method}

The \emph{quantum inverse scattering method} (QISM) was developed by the Faddeev school \cite{Faddeev1995, KorepinBI, Sklyanin1982, SklyaninFaddeev, SklyaninTakhtajanFaddeev} after Baxter's pioneering work on exactly solvable models in statistical mechanics and his method of commuting transfer matrices; see \cite{Baxter1989} for a textbook account and references therein.
It turns out \cite{Gutkin1988, ReedSimon2, Oxford} that the QNLS Hamiltonian can be expressed in terms of quantum field operators, certain operator-valued distributions, associated to a non-relativistic quantum field theory (the formalism known as ``second quantization''). 
The corresponding time evolution equation is a quantized version of the classical nonlinear Schr\"odinger equation, whence the name of our quantum model. Using the QISM the QNLS model can be solved, as follows. \\
\newpage

Consider the standard $\L^2$-inner product $\innerrnd{f}{g} = \int_{J^N} \dd^N \bm x f(\bm x) \overline{g(\bm x)}$. 
$S_N$ acts on the \emph{$N$-particle sector} $\f h_N = \L^2(J^N)$ by means of $(w f)(\bm x) = f(w^{-1} \bm x)$, with $\innerrnd{w f}{g} = \innerrnd{f}{w^{-1} g}$ for all $w \in S_N$. 
$\f h_N$ is a Hilbert space with respect to $\innerrnd{}{}$, as are $\ca H_N=\f h_N^{S_N}$ and $\ca H$, the \emph{(bosonic) Fock space}, the subset of the direct sum of all $\ca H_N$ consisting of elements of finite $\L^2$-norm.
The QISM revolves around the so-called \emph{monodromy matrix} \cite{Gutkin1988,KorepinBI,Sklyanin1982}  
\[ \ca T_\mu = \left( \atop{A_\mu}{C_\mu} \atop{B_\mu}{D_\mu} \right) \in \End(\C^2 \otimes \ca H), \]
which satisfies the \emph{Yang-Baxter equation} (YBE) related to the Yangian of $\f{gl}_2$, viz.
\begin{equation} \ca R_{\mu-\nu} (\ca T_\mu \otimes 1) (1 \otimes \ca T_\nu) = (1 \otimes \ca T_\nu)  (\ca T_\mu \otimes 1) \ca R_{\mu-\nu} \in \End(\C^2 \otimes \C^2 \otimes \ca H), \label{eqn:YBE}
\end{equation}
where the \emph{$R$-matrix} $\ca R_{\mu}$ is an element of $\End(\C^2 \otimes \C^2)$ with a meromorphic dependence on $\mu$.
For the QNLS model, the $R$-matrix is given by $\ca R_{\mu} = 1 - \frac{\ii \gamma}{\mu} \ca P$, with $\ca P$ the flip operator$: \C^2 \otimes \C^2 \to \C^2 \otimes \C^2: \bm v_1 \otimes \bm v_2 \mapsto \bm v_2 \otimes \bm v_1$.
The matrix entries $A_\mu, B_\mu , C_\mu, D_\mu $ can be explicitly given as integral operators \cite{Gutkin1988,KorepinBI,Sklyanin1982}.
For any model with the $R$-matrix given as above, the YBE \eqref{eqn:YBE} encodes commutation relations for the operators $A_\mu, B_\mu , C_\mu, D_\mu $:
\begin{equation} \label{eqn:ABCDcommrel}
\begin{array}{!{\vrule width 2pt}C{24mm}|C{24mm}|C{24mm}|C{24mm}!{\vrule width 2pt}}
\noalign{\hrule height 2pt}
[A_\mu,A_\nu ] =0 & [B_\mu ,B_\nu ] = 0  & [C_\mu ,C_\nu ] = 0 & [D_\mu,D_\nu] =0\\
\noalign{\hrule height 2pt}
\multicolumn{2}{!{\vrule width 2pt}>{\hspace{-1mm}}c<{\hspace{-1mm}}|}{\left[A_\mu,B_\nu \right] = \frac{- \ii \gamma}{\mu\!-\!\nu} \left( B_\mu A_\nu \! - \!B_\nu A_\mu \right)} &
\multicolumn{2}{>{\hspace{-1mm}}c<{\hspace{-1mm}}!{\vrule width 2pt}}{\left[A_\mu,C_\nu \right] = \frac{ \ii \gamma}{\mu\!-\!\nu} \left( C_\mu A_\nu \! - \!C_\nu A_\mu \right)} \\
\multicolumn{2}{!{\vrule width 2pt}>{\hspace{-1mm}}c<{\hspace{-1mm}}|}{\left[B_\mu ,A_\nu \right] = \frac{- \ii \gamma}{\mu\!-\!\nu} \left( A_\mu B_\nu \! - \!A_\nu B_\mu \right)} &
\multicolumn{2}{>{\hspace{-1mm}}c<{\hspace{-1mm}}!{\vrule width 2pt}}{\left[C_\mu ,A_\nu \right] = \frac{ \ii \gamma}{\mu\!-\!\nu} \left( A_\mu C_\nu \! - \!A_\nu C_\mu \right)}\\
\hline
\multicolumn{2}{!{\vrule width 2pt}>{\hspace{-1mm}}c<{\hspace{-1mm}}|}{\left[D_\mu,B_\nu \right] = \frac{\ii \gamma}{\mu\!-\!\nu} \left( B_\mu D_\nu \! - \!B_\nu D_\mu \right)} &
\multicolumn{2}{>{\hspace{-1mm}}c<{\hspace{-1mm}}!{\vrule width 2pt}}{\left[D_\mu,C_\nu \right] = \frac{-\ii \gamma}{\mu\!-\!\nu} \left( C_\mu D_\nu \! - \!C_\nu D_\mu \right)} \\
\multicolumn{2}{!{\vrule width 2pt}>{\hspace{-1mm}}c<{\hspace{-1mm}}|}{\left[B_\mu ,D_\nu \right] = \frac{ \ii \gamma}{\mu\!-\!\nu} \left( D_\mu B_\nu \! - \!D_\nu B_\mu \right)} &
\multicolumn{2}{>{\hspace{-1mm}}c<{\hspace{-1mm}}!{\vrule width 2pt}}{\left[C_\mu ,D_\nu \right] = \frac{-\ii \gamma}{\mu\!-\!\nu} \left( D_\mu C_\nu \! - \!D_\nu C_\mu \right)}\\
\noalign{\hrule height 2pt}
\multicolumn{2}{!{\vrule width 2pt}>{\hspace{-1mm}}c<{\hspace{-1mm}}|}{\left[A_\mu,D_\nu \right] = \frac{-\ii \gamma}{\mu\!-\!\nu} \left( B_\mu C_\nu\! - \!B_\nu C_\mu \right)} &
\multicolumn{2}{>{\hspace{-1mm}}c<{\hspace{-1mm}}!{\vrule width 2pt}}{\left[B_\mu ,C_\nu \right] = \frac{-\ii \gamma}{\mu\!-\!\nu} \left( A_\mu D_\nu\! - \!A_\nu D_\mu \right)}\\
\multicolumn{2}{!{\vrule width 2pt}>{\hspace{-1mm}}c<{\hspace{-1mm}}|}{\left[D_\mu ,A_\nu \right] = \frac{-\ii \gamma}{\mu\!-\!\nu} \left( C_\mu B_\nu\! - \!C_\nu B_\mu  \right)} &
\multicolumn{2}{>{\hspace{-1mm}}c<{\hspace{-1mm}} !{\vrule width 2pt}}{\left[C_\mu ,B_\nu \right] = \frac{-\ii \gamma}{\mu\!-\!\nu} \left( D_\mu A_\nu\! - \!D_\nu A_\mu\right)} \\
\noalign{\hrule height 2pt}
\end{array}
\end{equation}
The algebra so generated is called the \emph{Yang-Baxter algebra}.
The relevance of the monodromy matrix lies in the fact that the \emph{transfer matrices} $\Tr_{\C^2} \ca T_\mu = A_\mu+D_\mu $ form a self-adjoint commuting family and are generating functions for the integrals of motion, including the Hamiltonian $H^\gamma_{(N)}$ \cite{Faddeev1995,KorepinBI}. The eigenfunction $\Psi_{\bm \lambda}$ can be recursively generated, viz.
\begin{equation} \label{eqn:Psirecursion} \Psi_{\lambda_1,\ldots,\lambda_N} = B_{\lambda_N} \cdots B_{\lambda_1} \Vac, \end{equation}
where $\Vac = 1 \in  \ca H_0 \cong \C$.
In addition, for bounded $J$, if the BAEs \eqref{eqn:BAEintro} hold, $\Psi_{\bm \lambda}$ is an eigenfunction of the transfer matrix, and hence of $H^\gamma_{(N)}$. This method of constructing $\Psi$ is known as the \emph{algebraic Bethe ansatz} (ABA).

\subsection{Root system generalizations and recursive structure}

An important contribution by Gaudin \cite{Gaudin1971-3} was the realization that the BA approach can be modified to solve certain generalizations of the Lieb-Liniger system in terms of classical (crystallographic reduced) root systems, which have been the subject of further study \cite{Emsiz,EmsizOS,Gutkin1982,GutkinSutherland,HeckmanOpdam1997}. It has been highlighted by Heckman and Opdam \cite{HeckmanOpdam1997} that representations of a certain degeneration of the affine Hecke algebra play an essential role, providing another method for solving the QNLS problem, and in fact one which works for all root systems \cite{EmsizOS}. 
We will review the A-type case in Sect. \ref{sec:3}. \\

Unfortunately, there exists no generalization of the QISM to arbitrary root systems\footnote{However, the C-type analogues, both finite and affine, involving 1 or 2 reflecting boundaries, respectively, can be solved by Sklyanin's \emph{boundary Yang-Baxter equation} formalism \cite{Sklyanin1988}.}. 
On the other hand, the drawback of the dAHA method is the apparent lack of a recursive structure, which we will address in this paper. 
The results of this paper, in particular of Sect. \ref{sec:4}, are an indication of a deep connection between these two solution methods, at least for the A-type case. 
This interplay can be seen as something reminiscent of Schur-Weyl duality \cite{Schur1,Schur2,Weyl}; the Yangian of $\f{gl}_2$, the algebraic object underlying the YBE, is a deformation of the current algebra of $\f{gl}_2$ and its representation theory should be related to that of the degenerate affine Hecke algebra, which is a deformation of the group algebra of $S_N$.

\section{The degenerate affine Hecke algebra} \label{sec:3}

We review the existing theory of the so-called \emph{Dunkl-type operators} associated to the QNLS model, in particular the study of their eigenfunctions as a means to understanding the QNLS model as started by Komori and Hikami \cite{Hikami,KomoriHikami}. 
A running example for the case $N=2$ is provided in Appendix \ref{sec:example}.

\subsection{Dunkl-type operators; degenerate affine Hecke algebra}

Given $w \in S_N$, consider the set
\[ \Sigma_w = \set{(j,k) \in \{1,\ldots,N\}^2}{j<k, \, w(j)>w(k)}, \]
which labels those positive roots in the standard realization of the A$_{N-1}$-type root system which are mapped by $w$ to a negative root; hence the length of $w$ is given by $l(w) = |\Sigma_w|$.
Writing $w = s_{i_1} \cdots s_{i_l}$ with $l=l(w)$ we have
\begin{equation} \label{eqn:Sigmaproperty} \set{w s_{j \, k} }{(j,k) \in \Sigma_w} = \set{s_{i_1} \cdots \hat s_{i_m} \cdots s_{i_l} }{m=1,\ldots,l}, \end{equation}
where $\hat s_{i_m}$ indicates that $s_{i_m}$ is removed from the product.
This follows by induction on $l$ and the equivalence of the statements
\begin{equation} \label{eqn:Sigmaequivalence}
\begin{gathered}
l(w_1w_2) = l(w_1)+l(w_2),  \\
\Sigma_{w_1 w_2} = w_2^{-1} \Sigma_{w_1} \cup \Sigma_{w_2}, \\
\forall(j,k) \in w_2^{-1} \Sigma_{w_1}: \, j<k.
\end{gathered}
\end{equation}
for which see \cite[Eqn. (2.2.4)]{Macdonald2}.
For $j=1,\ldots,N$ we define the auxiliary operator $\Lambda_j\in \End(\ca C^\infty(\R^N_\text{reg}))$ by specifying its action on each alcove:
\begin{equation} \label{eqn:Lambdadefn} \Lambda_j f = \sum_{k: (k,j) \in \Sigma_w} s_{j \, k}f - \sum_{k: (j,k) \in \Sigma_w} s_{j \, k}f \qquad \text{on } w^{-1} \R^N_+,
\end{equation}
for $w \in S_N,f \in \ca C^\infty(\R^N_\text{reg})$. Note that $ (s_{j \, k}f)|_{w^{-1} \R^N_+} = s_{j \, k}(f|_{(ws_{j \, k})^{-1} \R^N_+})$ and for $(j,k) \in \Sigma_w$ we have $l(ws_{j \, k})<l(w)$ as per \eqref{eqn:Sigmaproperty}.

\begin{defn} \emph{\cite{Hikami,MurakamiWadati,Opdam}}
Let $j=1,\ldots,N$.
The \emph{Dunkl-type operator} $\partial_j^\gamma$ is given by
\begin{equation} 
\partial_j^\gamma = \partial_j - \gamma \Lambda_j \in \End(\ca C^\infty(\R^N_\text{reg})).
\end{equation}
In particular, for $f \in \ca{C}^\infty(\R^N_\text{reg})$ we have 
\begin{equation} \qquad \qquad \partial_j^\gamma f = \partial_j f \qquad \text{on } \R^N_+. \label{eqn:Dunklfundamentalalcove} \end{equation}
\end{defn}
It is well-established (see, e.g. \cite{MurakamiWadati,Opdam}) that the simple transpositions $s_j$ ($j=1,\ldots,N-1$) and the Dunkl-type operators $\partial_k^\gamma $ ($k=1,\ldots,N$) satisfy the following relations in $ \End(\ca C^\infty(\R^N_\text{reg}))$:
\begin{gather}
s_j^2 = 1, \qquad s_j s_{j+1} s_j = s_{j+1}  s_j s_{j+1} \qquad \text{and} \qquad s_j s_k = s_k s_j \; \text{for } |j-k|>1; \label{eqn:Dunklrel1}  \\
s_j \partial_j^\gamma - \partial_{j+1}^\gamma  s_j = \gamma \qquad \text{and} \qquad s_j \partial_k^\gamma  = \partial_k^\gamma  s_j  \; \text{for } k \ne j,j+1; \label{eqn:Dunklrel2}\\
\partial_j^\gamma \partial_k^\gamma  = \partial_k^\gamma  \partial_j^\gamma.   \label{eqn:Dunklrel3}
\end{gather}
This implies that $s_1,\ldots,s_{N\!-\!1},-\ii \partial_1^\gamma ,\ldots,-\ii \partial_N^\gamma $ define a representation (to be called the \emph{Dunkl-type representation}) of the \emph{degenerate affine Hecke algebra} (dAHA) $\ca{A}_N^\gamma$, as introduced by Drinfel\cprime d \cite{Drinfeld1986} and Lusztig and Kazhdan \cite{KazhdanLusztig,Lusztig}.
Note that $\ca{A}_N^\gamma \cong \C S_N \otimes \C[X_1,\ldots,X_N]$ as vector spaces; the algebra multiplication is a $\gamma$-deformation of the standard action of $S_N$ on the polynomial algebra, as per \eqref{eqn:Dunklrel2}.
The centre of $\ca{A}_N^\gamma$ is given by the symmetric expressions in the generators of its polynomial subalgebra \cite{Cherednik2, Lusztig, Opdam}. 
In particular, for $F \in \C[\lambda_1,\ldots,\lambda_N]^{S_N}$ and $w \in S_N$ we have $[w,F(\partial_1^\gamma ,\ldots,\partial_N^\gamma )] =0$.
Using \eqref{eqn:Dunklfundamentalalcove} we obtain
\begin{equation}\label{eqn:Dunklsymmetricpolynomial}
F(\partial_1^\gamma ,\ldots,\partial_N^\gamma ) = F(\partial_1,\ldots,\partial_N) \in \End(\ca C^\infty(\R^N_\text{reg})). \end{equation}

\subsection{Integral-reflection operators; the propagation operator}

For $1 \leq j \ne k \leq N$, consider the integral operator $I_{j \, k} \in \End\left(\ca C(\R^N)\right)$ defined by
\[ (I_{j \, k} f)(\bm x) = \int_0^{x_j-x_k} \dd y f \! \left( \bm x - y (\bm e_j-\bm e_k) \right) = \int_{x_k}^{x_j} \dd y f(x_1, \ldots, \underset{(j)}{x_j+x_k-y}, \ldots, \underset{(k)}{y},\ldots,x_N ) \]
for $f\in \ca C(\R^N)$ and $\bm x \in \R^N$; we note that $I_{j \, k}$ restricts to an operator on $\ca C^\infty(\R^N)$.
This operator was introduced as a tool to study the QNLS problem in \cite{Gutkin1982,GutkinSutherland}.
Given $1 \leq j \ne k \leq N$, we have 
\begin{gather}
w I_{j \, k} = I_{w(j) \, w(k)}w \quad \text{for } w\in S_N, \label{eqn:Ipermutation} \\
I_{j \, k}f = 0 \quad \text{on } V_{j \, k}, \text{ for } f \in \ca C(\R^N). \label{eqn:Ihyperplanes}
\end{gather} 
For $j=1,\ldots,N-1$, we write $I_j := I_{j \, j+1}$ and introduce the \emph{integral-reflection operator}
\begin{equation} \label{eqn:sgamma} s_j^\gamma:=s_j + \gamma I_j \in \End(\ca C(\R^N)), \End(\ca C^\infty(\R^N)).\end{equation}
It can be checked \cite{EmsizOS,HeckmanOpdam1997} that the $s_j^\gamma$ ($j=1,\ldots,N-1$) and $-\ii \partial_k$ ($k=1,\ldots,N$) define a representation, to be called the \emph{integral representation}, of the dAHA on $\End(\ca C^\infty(\R^N))$, i.e. we may replace $(s_j,\partial_k^\gamma ) \mapsto (s_j^\gamma, \partial_k)$ in \eqref{eqn:Dunklrel1}-\eqref{eqn:Dunklrel3}. 
Hence, given any $w \in S_N$ and any decomposition $w=s_{i_1} \cdots s_{i_l}$, the expression $s_{i_1}^\gamma  \cdots s_{i_l}^\gamma $ is independent of the choice of the $s_{i_m}$, and we will denote this element of $\End(\ca C^\infty(\R^N))$ by $w^\gamma$.\\

Following Hikami \cite{Hikami} we study the intertwiner of the aforementioned two representations, restricted to a suitable function space.

\begin{defn}
The \emph{propagation operator} is the element of $\End(\ca C(\R^N_\text{reg}))$ determined by
\begin{equation} \label{eqn:proprestricted}
P^\gamma_{(N)} f = w^{-1} w^\gamma f \qquad \text{on }  w^{-1} \R^N_+.
\end{equation}
\end{defn}
for $w \in S_N$, $f \in \ca C(\R^N_\text{reg})$. 
This operator was introduced by Gutkin \cite{Gutkin1982}.
From \eqref{eqn:Ihyperplanes} one obtains that $P^\gamma_{(N)}$ restricts to an element of $\End(\ca C(\R^N))$;
furthermore, since the $w^\gamma$ restrict to elements of $\End(\ca C^\infty(\R^N))$, the propagation operator $P^\gamma_{(N)}$ restricts to an element\footnote{Following \cite{EmsizOS} we repeat that $P^\gamma_{(N)}$ does not restrict to an element of $\End(\ca{CB}^\infty(\R^N))$.} of $\Hom(\ca C^\infty(\R^N),\ca{CB}^\infty(\R^N))$.
Crucially, $P^\gamma_{(N)}$ \emph{intertwines} the integral and Dunkl-type representations of the dAHA \cite{EmsizOS,Hikami}:
\begin{align}
w P^\gamma_{(N)} = P^\gamma_{(N)} w^\gamma \hspace{10mm} & \in \Hom(\ca C^\infty(\R^N),\ca{CB}^\infty(\R^N)), \label{eqn:intertwine1}\\
\partial_j^\gamma (P^\gamma_{(N)}|_{\R^N_\text{reg}}) = (P^\gamma_{(N)} \partial_j)|_{\R^N_\text{reg}} & \in \Hom(\ca C^\infty(\R^N),\ca C^\infty(\R^N_\text{reg})), \label{eqn:intertwine2}
\end{align}
for $w \in S_N$ and $j=1,\ldots,N$.
\eqref{eqn:intertwine1} is established straightforwardly by making a variable substitution in the summation in $P^\gamma_{(N)}$.
\eqref{eqn:intertwine2} is shown on each alcove $w^{-1} \R^N_+$, where one uses \eqref{eqn:intertwine1} and the fact that $(w^\gamma,-\ii \partial_j)$ defines a representation of the dAHA, so that we may use the well-known identity \cite{Cherednik2,Lusztig,Opdam}
\[ w^\gamma \partial_j = \partial_{w(j)} w^\gamma + \gamma w^\gamma \Biggl( \sum_{k: (j,k) \in \Sigma_w} s_{j \, k}^\gamma - \sum_{k: (k,j) \in \Sigma_w} s_{j \, k}^\gamma \Biggr).  \]
From $ s_{j \, k}I_{j \, k} = I_{k \, j}s_{j \, k} = I_{k \, j} = -I_{j \, k}$ it follows that $s_j s_j^\gamma = 1+\gamma I_{j+1 \, j}$; combining this with \eqref{eqn:Ipermutation} we obtain the following useful expression for $P^\gamma_{(N)}|_{s_{N-1} \cdots s_m \R^N_+}$:
\begin{equation}
\label{eqn:proprestricted2}
P^\gamma_{(N)} = s_{N \!- \!1} \cdots s_m s_{m}^\gamma  \cdots s_{N\!-\!1}^\gamma   = \bigl( 1+\gamma I_{N \, m} \bigr) \cdots \bigl( 1+\gamma I_{N \, N\!-\!1} \bigr) 
\quad \text{on } s_{N-1} \cdots s_m \R^N_+ .
\end{equation}

\subsection{Common eigenfunctions of the Dunkl-type operators}

Let $j=1,\ldots,N$ and $f \in \ca{CB}^\infty(\R^N)$; suppose that $\partial_j^\gamma f|_{\R^N_\text{reg}} \in \ca{C}^\infty(\R^N_\text{reg})$ is a constant multiple of $f|_{\R^N_\text{reg}}$.
Hence $\partial_j^\gamma f|_{\R^N_\text{reg}}$ can be continuously extended to $\R^N$ and seen as an element of $\ca{CB}^\infty(\R^N)$.
Therefore, given $\bm \lambda = (\lambda_1,\ldots,\lambda_N) \in \C^N$ the system $\{ \partial_j^\gamma f = \ii \lambda_j f \}_{j=1}^N$ for $f \in \ca{CB}^\infty(\R^N)$ is well-posed and we have

\begin{defn}
Let $\bm \lambda \in \C^N$ and $\gamma \in \R$.
Then we have the following subspace of $\ca{CB}^\infty(\R^N)$:
\begin{equation}
\text{sol}^\gamma_{\bm \lambda} = \set{f \in \ca{CB}^\infty(\R^N)}{\partial_j^\gamma f = \ii \lambda_j f \text{ for } j=1,\ldots,N}
\end{equation}
\end{defn}

Two technical lemmas follow.
\begin{lem} \label{lem:Dunklsystemzero1}
Let $\bm \lambda \in \C^N$ and $w \in S_N$. Suppose that $f \in \text{sol}_{\bm \lambda}^\gamma$ satisfies $f(\bm 0)=0$. 
If $f$ vanishes on ${\tilde w}^{-1} \R^N_+$ for all $\tilde w \in S_N$ with $l(\tilde w)<l(w)$, then $f$ vanishes on $w^{-1}\R^N_+$.
\end{lem}

\begin{proof}
Assume that for all $\tilde w \in S_N$ with $l(\tilde w)<l(w)$ we have $f=0$ on ${\tilde w}^{-1} \R^N_+$.
From \eqref{eqn:Lambdadefn} and the subsequent comments we conclude that $ \Lambda_j f = 0$ on ${\tilde w}^{-1} \R^N_+$.
Hence $\partial_j (f|_{w^{-1} \R^N_+})=\partial_j^\gamma f|_{w^{-1} \R^N_+}= \ii \lambda_j f|_{w^{-1} \R^N_+}$ for $j=1,\ldots,N$ so that $f = c_w \e^{\ii \bm \lambda}$ on $w^{-1} \R^N_+$ for some $c_w \in \C$.
Continuity at $\bm x = \bm 0$ yields that $f=0$ on $w^{-1} \R^N_+$.
\end{proof}

\begin{lem} \label{lem:Dunklsystemzero2}
Let $\bm \lambda \in \C^N$. Suppose that $f \in \text{sol}_{\bm \lambda}^\gamma$ satisfies $f(\bm 0)=0$.
Then $f=0$.
\end{lem}

\begin{proof}
First we will establish that $f=0$ on all alcoves $w^{-1} \R^N_+$ by induction on $l(w)$; \rfl{lem:Dunklsystemzero1} implies both the base case $l(w) = 0$ (where there are no $\tilde w \in S_N$ for which $l(\tilde w)<l(w)$) and the induction step. Hence $f=0$ on the regular vectors and by continuity we have $f=0$.
\end{proof}

Note that $\text{sol}_0(\bm \lambda)$ is 1-dimensional, and spanned by $\e^{\ii \bm \lambda}$. 
Something similar holds for general $\gamma$.

\begin{prop}
\label{prop:Dunklsystemonedimensional}
Let $\bm \lambda \in \C^N$. 
$\text{sol}_{\bm \lambda}^\gamma$ is 1-dimensional and spanned by $P^\gamma_{(N)} \e^{\ii \bm \lambda}$.
\end{prop}

\begin{proof}
Suppose that $f,g \in \text{sol}_{\bm \lambda}^\gamma$ and $f \ne 0$, $g \ne 0$. We wish to show that $f$ is a multiple of $g$. From \rfl{lem:Dunklsystemzero2} we conclude that $f(\bm 0) \ne 0 \ne g(\bm 0)$ and the function $\tilde g := \frac{f(\bm 0)}{g(\bm 0)} g$ satisfies $f(\bm 0) = \tilde g(\bm 0)$. 
Note that $h=f- \tilde g$ is also an element of $\text{sol}_{\bm \lambda}^\gamma$, and $h(\bm 0)=0$. From \rfl{lem:Dunklsystemzero2} it follows that $h=0$; hence $f$ is a multiple of $g$. $P^\gamma_{(N)} \e^{\ii \bm \lambda} \in \text{sol}_{\bm \lambda}^\gamma$ follows from the intertwining property \eqref{eqn:intertwine2}; it is nonzero since $P^\gamma_{(N)} \e^{\ii \bm \lambda}(\bm 0) = \lim_{\R^N_+ \ni \bm x \to \bm 0} \e^{\ii \bm \lambda}(\bm x) \ne 0$.
\end{proof}

\begin{defn}
Let $\bm \lambda \in \C^N$.
The corresponding \emph{non-symmetric wavefunction} is defined as
\[ \psi_{\bm \lambda}  = P^\gamma_{(N)} \e^{\ii \bm \lambda} \in \ca{CB}^\infty(\R^N). \]
\end{defn}

\begin{lem} \label{lem:Weylgponpsi}
Let $\bm \lambda \in \C^N$.
For $j=1,\ldots,N-1$ we have
\begin{equation} \label{eqn:permutednonsymmwavefn}
s_j \psi_{\bm \lambda} = \psi_{s_j \bm \lambda} - \ii \gamma\frac{\psi_{\bm \lambda} - \psi_{s_j \bm \lambda}}{\lambda_j-\lambda_{j+1}}. \end{equation}
If $\lambda_j=\lambda_{j+1}$ the right-hand side is to be interpreted as a limit: $\lambda_{j+1} \to \lambda_j$.
\end{lem}

\begin{proof}
Straightforwardly one finds that $I_j \e^{\ii \bm \lambda} = \frac{-\ii}{\lambda_j-\lambda_{j+1}} \left(\e^{\ii {\bm \lambda}} - \e^{\ii {s_j \bm \lambda}} \right)$.
Using this and the intertwining property \eqref{eqn:intertwine1} we obtain the lemma. \end{proof}

Let $N>1$ and $j=1,\ldots,N-1$.  \eqref{eqn:permutednonsymmwavefn} implies that unless $\lambda_j-\lambda_{j+1}=-\ii \gamma$ we see that $\psi_{\bm \lambda} \ne s_j \psi_{\bm \lambda}$.
It follows that for generic $\bm \lambda$, $\psi_{\bm \lambda}$ is not $S_N$-invariant.

\begin{defn}
Let $\bm \lambda \in \C^N$.
The corresponding \emph{symmetric wavefunction} is given by
\[ \Psi_{\bm \lambda} = \ca S_{(N)}  \psi_{\bm \lambda} \in \ca{CB}^\infty(\R^N)^{S_N}. \] 
\end{defn}

\subsection{Connection to the QNLS model}

The relevance of the space $\text{sol}_{\bm \lambda}^\gamma$ to the QNLS problem is expressed in
\begin{prop}  \label{prop:DunklQNLS}
Suppose that $f \in \text{sol}_{\bm \lambda}^\gamma$ for some $\bm \lambda \in \C^N$.
Then $f$ solves \eqref{eqn:QNLSLaplace}-\eqref{eqn:QNLSjump} with $E = \sum_j \lambda_j^2$, i.e. it is a non-symmetric solution to the QNLS eigenvalue problem.
\end{prop}

\begin{proof}
\eqref{eqn:QNLSLaplace} follows from the fact that $\sum_{j=1}^N \partial_j^2 $ and $ \sum_{j=1}^N (\partial_j^\gamma)^2$ are the same on $\R^N_\text{reg}$, which is a consequence of \eqref{eqn:Dunklsymmetricpolynomial} with $F$ equal to the sum-of-squares polynomial. 
As for the derivative jump conditions \eqref{eqn:QNLSjump}, consider the hyperplane $V_{j \,k}$ with $1 \leq j < k \leq N$. 
Let $\bm x \in V_{j\,k}$ be subregular with $\bm x_\delta \in w^{-1} \R^N_+$, say, for $\delta>0$ small enough.
We know that $j=w^{-1}(r)$ and $k=w^{-1}(r+1)$ for some $r=1,\ldots,N-1$ and hence $\Sigma_{s_r w}= \Sigma_w \cup \{ (j,k) \}$ by \eqref{eqn:Sigmaequivalence}. 
From $\bm x_{-\delta} \in s_{j\,k}w^{-1} \R^N_+ = (s_rw)^{-1} \R^N_+$ we have
\begin{align*}
&\lim_{\delta \to 0} \left( (\Lambda_j f)(\bm x_\delta) - (\Lambda_j f)(\bm x_{-\delta})  \right) =  \\
&= \lim_{\delta \to 0} \Bigl( \sum_{l:(l,j) \in \Sigma_w} \! (s_{j \, l}f)(\bm x_\delta) - \! \sum_{l:(j,l) \in \Sigma_w} \! (s_{j \, l}f)(\bm x_\delta) -\! \sum_{l:(l,j) \in \Sigma_{s_r w}} \! (s_{j \, l}f)(\bm x_{-\delta})  + \! \sum_{l:(j,l) \in \Sigma_{s_r w}} \! (s_{j \, l}f)(\bm x_{-\delta})  \Bigr)   \\
&= \lim_{\delta \to 0} \Bigl( \sum_{l:(l,j) \in \Sigma_w} \! \left( (s_{j \, l}f)(\bm x_\delta) - (s_{j \, l}f)(\bm x_{-\delta})  \right) - \! \sum_{l:(j,l) \in \Sigma_w}\! \left( (s_{j \, l}f)(\bm x_\delta) - (s_{j \, l}f)(\bm x_{-\delta})  \right) 
+ (s_{j \, k} f)(\bm x_{-\delta}) \Bigr) \\
&= f(\bm x),
\end{align*}
since all $s_{j\,l} f \in \ca C(\R^N)$. 
Similarly, $ \lim_{\delta \to 0} \left( (\Lambda_k f)(\bm x_\delta) - (\Lambda_k f)(\bm x_{-\delta})  \right) = -f(\bm x)$. 
\eqref{eqn:QNLSjump} now follows from $\partial_j f = \ii \lambda_j f + \gamma \Lambda_j f$ and the continuity of $f$.
\end{proof}

By virtue of \rfp{prop:Dunklsystemonedimensional} we have
\begin{prop} \label{prop:psiQNLS}
For all $\bm \lambda \in \C^N$, $\psi_{\bm \lambda}$ satisfies \eqref{eqn:QNLSLaplace}-\eqref{eqn:QNLSjump} with $E = \sum_j \lambda_j^2$.
\end{prop}

From \rfp{prop:DunklQNLS} and \rfp{prop:psiQNLS} we obtain
\begin{cor}
Let $\bm \lambda \in \C^N$.
Then $\Psi_{\bm \lambda}$ satisfies the derivative jump conditions \eqref{eqn:QNLSjump}.
Furthermore, for any symmetric polynomial $F \in \C[\bm \lambda]^{S_N}$, $\Psi_{\bm \lambda}$ is an eigenfunction of $F(\partial_1^\gamma ,\ldots,\partial_N^\gamma )$ with eigenvalue $F(\ii \bm \lambda)$; in particular, $\Psi_{\bm \lambda}$ solves \eqref{eqn:QNLSLaplace} with $E=\sum_j \lambda_j^2$. Hence, $\Psi_{\bm \lambda}$ is an eigenfunction of the QNLS Hamiltonian $H^\gamma_{(N)}$.
\end{cor}

\begin{rem}
Symmetric polynomials in the $\partial_j^\gamma$ can be interpreted as constants of motion.
\end{rem}

%\begin{rem}
%There also exists a large body of theory dealing with the fermionic QNLS model (e.g. \cite{MurakamiWadati,MurakamiWadati2,Yang1967}). Without going in detail, we remark that the physically relevant wavefunctions can in principle be constructed from the non-symmetric wavefunctions discussed in this article by antisymmetrizing, i.e. by calculating $\frac{1}{N!} \sum_{w \in S_N} \sgn(w) \psi(x_{w 1},\ldots,x_{w N})$.
%\end{rem}

Assume $J$ is bounded. 
We are interested under what conditions the restriction of the function $\Psi_{\bm \lambda}$ to $J^N$ can be extended to an $L$-periodic function on $\R^N$, i.e. invariant with respect to the translation group of the lattice $L \Z^N$.
The following statement is well-known \cite{KorepinBI,LiebLi} and can be straightforwardly checked.
\begin{prop}
Let $\bm \lambda \in \C^N$. 
For $\Psi_{\bm \lambda}|_{J^N}$ to be able to be extended to a function which is $L$-periodic in each argument, continuous and smooth away from the set of affine hyperplanes 
\[ \left\{ V_{j,k;m}:= \set{\bm x \in \R^N_+}{x_j-x_k=m L} \right\}_{1 \leq j < k \leq N, m \in \Z} \]
it is necessary that the BAEs \eqref{eqn:BAEintro} are satisfied.
\end{prop}

Contrary to both the symmetric and the non-interacting ($\gamma=0$) case, by imposing conditions on $\bm \lambda$ one cannot extend the restricted non-symmetric wavefunction $\psi_{\bm \lambda}|_{J^N}$ to a function on $\R^N$ which is $L$-periodic in each variable.
We will demonstrate this in appendix \ref{sec:example} for $N=2$.\\

We remark that Dunkl-type operators and integral-reflection operators, and hence also the propagation operator can be defined in the affine setting \cite{EmsizOS} as well. 
For the A-type case, this means that the propagation operator generates a common eigenfunction of the Dunkl-type operators which satisfies the derivative jump conditions associated to the affine hyperplanes $V_{j \, k;m}$, but which itself is not invariant with respect to the action of the affine Weyl group $\hat S_N$, i.e. not symmetric and not invariant with respect to translations of the lattice $L \Z^N$. 

\section{Non-symmetric creation operators} \label{sec:4}

We will describe certain integral operators $b^\pm_\mu \in \Hom(\ca{CB}^\infty(\R^N),\ca{CB}^\infty(\R^{N\!+\!1}))$ that generate the non-symmetric wavefunctions recursively in two ways, viz.
\begin{equation} \label{eqn:psirecursion0}
\psi_{\lambda_1,\ldots,\lambda_{N+1}} = b^-_{\lambda_{N+1}} \psi_{\lambda_1,\ldots,\lambda_N} = b^+_{\lambda_1} \psi_{\lambda_2,\ldots,\lambda_{N+1}}
\end{equation}
corresponding to adding a particle to an $N$-particle system from the left ($b^+_\mu$) or the right ($b^-_\mu$). 
We remark upon the similarity of \eqref{eqn:Psirecursion} and \eqref{eqn:psirecursion0}, although in the former there is only one recursion, because the functions $\Psi_{\bm \lambda}$ are also invariant with respect to permuting the $\lambda_j$.

\subsection{Notations}

We will from now on reserve the symbol $n$ for the number of particle locations used as integration limits in certain integral operators. 
To keep track of which particle locations are used in these operators, we introduce the following sets of tuples of particle labels. 
Given an interval $K \subset \R$ and a nonnegative integer $n$ we introduce
\begin{align*}
\f i^n_K &= \set{(i_1,\ldots,i_n) \in (\Z \cap K)^n}{ i_l \ne i_m \text{ for } l \ne m},  \\
\f I^n_K &= \set{(i_1,\ldots,i_n) \in (\Z \cap K)^n}{i_1<\ldots<i_n} \subset \f i^n_K.
\end{align*}
Note that $\f i^0_K = \f I^0_K = \{ () \}$ and $\f i^n_K = \f I^n_K = \emptyset$ if $n$ exceeds the number of integers in $K$.
In the particular case $K=[1,N]$ we note that there is a faithful action of $S_N$ on $\f i^n_{[1,N]}$ in such a way that $\f I^n_{[1,N]}$ intersects each orbit in a point. \\

Given $1 \leq j < k \leq N$, define the \emph{step operator} $\theta_{j \, k} \in \End(\ca F(\R^N))$ by
\[ (\theta_{j \, k} f)(\bm x) = \begin{cases} f(\bm x), & x_j>x_k, \\ 0, & \text{otherwise}, \end{cases} \]
which restricts to an endomorphism of $\ca C^\infty(\R^N_\text{reg})$.
Extending this to $n$-tuples, given $\bm i \in \f i^n_{[1,N]}$, $n=1,\ldots,N$, define
\[ \theta_{\bm i} = \theta_{i_1 \, \ldots \, i_n} := \theta_{i_1 \, i_2} \cdots \theta_{i_{n-1} \, i_n} \in \End(\ca F(\R^N)), \]
We remark that $\theta_{\bm i}$ satisfies $w \theta_{\bm i}  = \theta_{w \bm i} w$ for $w \in S_N$. 
It maps $\ca C^\infty(\R^N_\text{reg})$ to itself and $\ca{CB}^\infty(\R^N)$ to $\set{f \in \ca F(\R^N)}{f|_{\R^N_\text{reg}} \in \ca C^\infty(\R^N_\text{reg}) }$.\\

It is convenient to introduce notation that respects the left-right distinction present in the recursions in \eqref{eqn:psirecursion0} as follows.
%; this will apply to (tuples of) particle labels as well as operators acting on subspaces of $\ca F(\R^N_{(\text{reg})})$, corresponding to the two obvious embeddings: $\ca F(\R^N_{(\text{reg})}) \to \ca F(\R^{N\!+\!1}_{(\text{reg})})$.
%In the same vein, to an operator acting on a subspace of $\ca F(\R^N_{(\text{reg})})$ we will associate an operator acting on the two subspaces of $\ca F(\R^{N\!+\!1}_{(\text{reg})})$ resulting from the two obvious embeddings: $\ca F(\R^N_{(\text{reg})}) \to \ca F(\R^{N\!+\!1}_{(\text{reg})})$. 
\begin{itemize}
\item We will write $j^+ = j+1 \in \{1,\ldots,N+1\}$ given $j \in \{0,\ldots,N\}$ and, where convenient, $j^- = j \in \{1,\ldots,N+1\}$ given $j \in \{1,\ldots,N+1\}$. 
Extending this to $n$-tuples, write $\bm i^+ = (i_1+1,\ldots,i_n+1) \in \f i^n_{[1,N\!+\!1]}$ for $\bm i = (i_1,\ldots,i_n) \in \f i^n_{[0,N]}$; equally, write $\bm i^- = \bm i$ for $\bm i \in \f i^n_{[1,N\!+\!1]}$.
\item For $\bm i \in \f i^n_{[0,N]}$, write $\theta_{\bm i}^+ := \theta_{\bm i^+} \in \End(\ca F(\R^{N\!+\!1}))$.
For $\bm i \in \f i^n_{[1,N\!+\!1]}$, write $\theta_{\bm i}^- := \theta_{\bm i^-} \in \End(\ca F(\R^{N\!+\!1}))$.
\item For $\epsilon = \pm$ and $j=1,\ldots,N-1$, write $s^\epsilon_j := s_{j^\epsilon} \in S_{N+1}$.
Extending this to all of $S_N$ by writing permutations as compositions of simple transpositions, given $w \in S_N$ we obtain $w^\epsilon \in S_{N+1}$ determined by $w^\epsilon(j^\epsilon)=(w(j))^\epsilon$ for $j=1,\ldots, N$, $w^+(1)=1$ and $w^-(N+1)=N+1$.
Using the integral representation of $\ca{A}_{N+1}^\gamma$ we also write $w^{\epsilon,\gamma}:=(w^\epsilon)^\gamma \in \End(\ca C(\R^{N\!+\!1}))$ for $w \in S_N$.
\item For $\epsilon = \pm$ and $j =1,\ldots,N$, define $\partial_j^{\epsilon(,\gamma)} := \partial_{j^\epsilon}^{(\gamma)} \in \End(\ca C^\infty(\R^{N\!+\!1}_\text{(reg)}))$.
\end{itemize}

Let $j=1,\ldots,N+1$. 
The assignment 
\[ (\hat \phi_j f)(\bm x) = f(x_1,\ldots,x_{j-1},x_{j+1},\ldots,x_{N+1}), \qquad \text{for } f \in \ca F(\R^N),  \,  \bm x \in \R^{N\!+\!1}, \]
defines basic particle creation operators $\hat \phi_j \in \Hom(\ca F(\R^N),\ca F(\R^{N\!+\!1}))$.
The $\hat \phi_j$ preserve smoothness on the alcoves and continuity, and hence restrict to elements of $\Hom(\ca{CB}^\infty(\R^N),\ca{CB}^\infty(\R^{N\!+\!1}))$.
Also, let $j=1,\ldots,N$ and $y \in \R$. 
The assignment 
\[ (\bar \phi_j(y)f)(\bm x) = f(x_1, \ldots, x_{j-1}, y, x_{j+1},\ldots, x_N), \qquad \text{for } f \in \ca F(\R^N), \bm x \in \R^N, \]
defines $\bar \phi_j(y) \in \End(\ca F(\R^N))$; it preserves continuity and smoothness on the alcoves, and hence restricts to an element of $\End(\ca{CB}^\infty(\R^N))$.
Given $n=0,\ldots,N$, $\bm i \in \f i^n_{[1,N]}$ and $\bm y \in \R^n$, in a multivariate setting we write
\[ \bar \phi_{\bm i}(\bm y)= \prod_{m=1}^n \bar \phi_{i_m}(y_m).  \]
We have
\[ w \bar \phi_{\bm i}(\bm y) = \bar \phi_{w \bm i}(\bm y) w \qquad \text{for } w \in S_N. \]

Furthermore, if the particle number is clear from the context we will write only a sign, and not an index, on an operator that is associated to the left- and rightmost variable in the following situations.
\begin{itemize}
\item We denote $s^+ = s_1, s^- = s_N \in S_{N+1}$, the simple transpositions acting on the first two and last two indices, respectively.
\item We write $\hat \phi^+ = \hat \phi_1$, $\hat \phi^- = \hat \phi_{N+1}$ $\in$ $\Hom(\ca F(\R^N),\ca F(\R^{N\!+\!1}))$. 
\end{itemize}
For $\epsilon = \pm$ we have $\hat \phi^\epsilon \bar \phi_{\bm i}(\bm y) = \bar \phi_{\bm i^\epsilon}(\bm y) \hat \phi^\epsilon$. 

\subsection{The operators $\hat e^\pm_{\mu ; \bm i}$ and $b^\pm_\mu$}

\begin{defn} 
Let $\bm i \in \f i^n_{[1,N]}$ and $\mu \in \C$.
The \emph{elementary (non-symmetric) creation operators} $ \hat e^\pm_{\mu;\bm i} \in \Hom(\ca{C}(\R^N),\ca{C}(\R^{N\!+\!1}))$ are defined by
\begin{align*}
( \hat e^+_{\mu;\bm i} f)(\bm x) &=  \e^{\ii \mu x_1} \theta_{\bm i \, 0}^+(\bm x) \left( \prod_{m=1}^n  \int_{x_{i_{m+1}^+}}^{x_{i_m^+}} \dd y_m \e^{\ii \mu (x_{i_m^+}-y_m)} \right) (\hat \phi^+ \bar \phi_{\bm i}(\bm y)f)(\bm x) , \\
( \hat e^-_{\mu; \bm i} f)(\bm x) &= \e^{\ii \mu x_{N+1}} \theta_{N+1 \, \bm i}^-(\bm x) \left( \prod_{m=1}^n \int_{x_{i_m^-}}^{x_{i_{m-1}^-}} \dd y_m \e^{\ii \mu (x_{i_m^-} -y_m)} \right) (\hat \phi^- \bar \phi_{\bm i}(\bm y)f)(\bm x) ,
\end{align*}
for $f \in \ca C(\R^N)$ and $\bm x \in \R^{N+1}$, where $i_{n+1}^+ = 1$ and $i_0^- =N+1$.
\end{defn}

Given $\epsilon = \pm$, $\hat e^\epsilon_{\mu;\bm i}$ preserves smoothness on the alcoves and vanishes continuously at the hyperplanes $V_{i_l^\epsilon \, i_m^\epsilon} \subset \R^{N\!+\!1}$ for $1 \leq l < m \leq n+1$ and $0 \leq l < m \leq n$, respectively.
Hence both $\hat e^\pm_{\mu;\bm i} \in \Hom(\ca{CB}^\infty(\R^N),\ca{CB}^\infty(\R^{N\!+\!1}))$. 
Moreover, by restricting the arguments of the functions on which the $\hat e^\pm_\mu$ act to regular vectors, we may view $\hat e^\pm_\mu \in \Hom(\ca{C}^\infty(\R^N_\text{reg}),\ca{C}^\infty(\R^{N\!+\!1}_\text{reg}))$.
For $w \in S_N$ and $\epsilon = \pm$ we have
\begin{equation} \label{eqn:epermutation}
w^\epsilon  \hat e^\epsilon_{\mu;\bm i} =  \hat e^\epsilon_{\mu;w \bm i}w.
\end{equation}

Note that for $n=0$ and $\epsilon=\pm$ the definitions $ \hat e^\epsilon_{\mu; \emptyset}$ simplify to
\begin{equation} \label{eqn:e0} \begin{aligned}
( \hat e^+_{\mu;\emptyset} f)(x_1,\ldots,x_{N+1}) &= \e^{\ii \mu x_1} f(x_2,\ldots,x_{N+1})\\
( \hat e^-_{\mu;\emptyset} f)(x_1,\ldots,x_{N+1}) &= \e^{\ii \mu x_{N+1}} f(x_1,\ldots,x_N),
\end{aligned} \end{equation}
which restrict to elements of $\Hom(\ca C^\infty(\R^N),\ca C^\infty(\R^{N\!+\!1}))$.
In the case $\gamma = 0$, these operators equal $b^\epsilon_\mu$ and the obvious statements $\hat e^+_{\mu;\emptyset} \e^{\ii \bm \lambda} = \e^{\ii (\mu,\bm \lambda)}$, $ \hat e^-_{\mu;\emptyset} \e^{\ii \bm \lambda} = \e^{\ii (\bm \lambda, \mu)}$ yield \eqref{eqn:psirecursion0}.

\begin{lem}[Relation between $\hat e^+_{\mu;\bm i}$ and $\hat e^-_{\mu;\bm i}$] \label{lem:eplusemin}
Let $\bm i \in \f i^n_{[1,N]}$.
We have
\[ \hat e^+_{\mu;\bm i} = 
s_{1 \, 2} s_{2 \, 3} \ldots s_{N \, N\!+\!1} s_{i_1 \, i_2} s_{i_2 \, i_3} \ldots s_{i_n \, N\!+\!1} \hat e^-_{\mu;\bm i} . \]
\end{lem}

\begin{proof}
By induction on $n$. For $n=0$, the statement yields $\hat e^+_{\mu;\emptyset} = s_{1 \, 2} s_{2 \, 3} \ldots s_{N \, N\!+\!1} \hat e^-_{\mu;\emptyset}$, which is obvious from \eqref{eqn:e0}. The induction step follows from the claim that given $\bm i \in \f i^n_{[1,N]}$ and $w \in S_{N+1}$, if $\hat e^+_{\mu;\bm i'} = w \hat e^-_{\mu;\bm i'}$ then $\hat e^+_{\mu;\bm i} = w s_{i_{n-1} \, i_n} \hat e^-_{\mu;\bm i}$, where $\bm i'=(i_1,\ldots,i_{n-1})$, which is easily established.
\end{proof}

\begin{defn}
Let $\mu \in \C$ and $\epsilon = \pm$. We define the \emph{non-symmetric creation operators}
\[ b^\epsilon_\mu  =  \sum_{n = 0}^N \gamma^n \sum_{\bm i \in \f i^{n}_{[1,N]}}  \hat e^\epsilon_{\mu;\bm i} \in \Hom(\ca{CB}^\infty(\R^N),\ca{CB}^\infty(\R^{N\!+\!1})). \]
\end{defn}
We may also think of $b^\pm_\mu$ as an element of $\Hom(\ca{C}^\infty(\R^N_\text{reg}),\ca{C}^\infty(\R^{N\!+\!1}_\text{reg}))$ or $\Hom(\ca C(\R^N),\ca C(\R^{N+1}))$.
From \eqref{eqn:epermutation} it follows that for $\mu \in \C$, $w \in S_N$ and $\epsilon = \pm$ we have
\begin{equation} \label{eqn:bwcommrel}  w^\epsilon b^\epsilon_{\mu} = b^\epsilon_\mu w \end{equation}

\subsection{Recursion of the $\psi_{\bm \lambda}$}

We now arrive at the heart of this paper, where we obtain the QISM-type recurrence relations \eqref{eqn:psirecursion0} for the $\psi_{\bm \lambda}$. 
One way of doing this is to establish certain commutation relations between $b^\pm_\mu$ and $\partial_j^\gamma$.
More precisely, for $\mu \in \C$ and $\epsilon = \pm$ we have
\begin{gather} 
\qquad \qquad \qquad \qquad \partial_j^{\epsilon,\gamma}  b^\epsilon_{\mu} = b^\epsilon_{\mu} \partial_j^\gamma \qquad \text{for } j=1,\ldots,N. \label{eqn:Dunklbcommrel2}\\
\partial_{N+1}^\gamma  b^-_\mu = \ii \mu b^-_\mu, \quad \partial_1^\gamma  b^+_\mu = \ii \mu b^+_\mu, \label{eqn:Dunklbcommrel1} 
\end{gather}
in $\Hom(\ca C^\infty(\R^N_\text{reg}),\ca C^\infty(\R^{N\!+\!1}_\text{reg}))$. 
For a proof of these identities we refer to \cite[App. B.2]{VlaarThesis}. 
Note that $b^-_{\lambda_N}  \cdots b^-_{\lambda_1} \Vac, b^+_{\lambda_1} \ldots b^+_{\lambda_N} \Vac \in \ca{CB}^\infty(\R^N)$ since $\Vac \in \ca{CB}^\infty(\R^0) \cong \C$; one establishes that $b^-_{\lambda_N}  \cdots b^-_{\lambda_1} \Vac, b^+_{\lambda_1} \ldots b^+_{\lambda_N} \Vac \in \text{sol}_{\bm \lambda}^\gamma$ as follows:
\begin{gather*} 
\partial_j^\gamma b^-_{\lambda_N}  \cdots b^-_{\lambda_1} \Vac = b^-_{\lambda_N} \partial_j^\gamma b^-_{\lambda_{N\!-\!1}} \cdots b^-_{\lambda_1} \Vac = \ldots = b^-_{\lambda_N}  \cdots b^-_{\lambda_{j\!+\!1}} \partial_j^\gamma b^-_{\lambda_j} \cdots b^-_{\lambda_1} \Vac = \ii \lambda_j b^-_{\lambda_N}  \cdots b^-_{\lambda_1} \Vac \\
\partial_j^\gamma b^+_{\lambda_1}  \cdots b^+_{\lambda_N} \Vac = b^+_{\lambda_1} \partial_{j\!-\!1}^\gamma b^+_{\lambda_2} \cdots b^+_{\lambda_N} \Vac = \ldots = b^+_{\lambda_1}  \cdots b^+_{\lambda_{j\!-\!1}} \partial_1^\gamma  b^+_{\lambda_j} \cdots b^+_{\lambda_N} \Vac = \ii \lambda_j b^+_{\lambda_1}  \cdots b^+_{\lambda_N} \Vac 
\end{gather*}
by virtue of \eqref{eqn:Dunklbcommrel2} and \eqref{eqn:Dunklbcommrel1}, as required. 
In light of \rfp{prop:Dunklsystemonedimensional}, it follows that $b^-_{\lambda_N}  \cdots b^-_{\lambda_1} \Vac$ and $b^+_{\lambda_1}  \cdots b^+_{\lambda_N} \Vac$ are multiples of $\psi_{\lambda_1,\ldots,\lambda_N}$.
To obtain equality, it suffices to show that they coincide on ${\R}^N_+$.
This follows from \eqref{eqn:proprestricted} and $b^\epsilon_{\mu} f|_{\R^n_+} = \hat e^\epsilon_{\mu;\emptyset} f|_{\R^n_+}$
for $n=0,\ldots,N-1$ and $f \in \ca{CB}^\infty(\R^n)$. 

\begin{rem}
\eqref{eqn:bwcommrel} and \eqref{eqn:Dunklbcommrel2} together express that the $b^\epsilon_\mu$ intertwine the Dunkl-type representation of $\ca{A}_N^\gamma$ with subrepresentations of the Dunkl-type representation of $\ca{A}_{N+1}^\gamma$.
\end{rem}

However, we will derive \eqref{eqn:psirecursion0} in a different way. 
\begin{thm} \label{thm:bpropcommrel}
Let $\mu \in \C$ and $\epsilon = \pm$. In $\Hom(\ca C^\infty(\R^N),\ca{CB}^\infty(\R^{N+1}))$ we have
\[ P^\gamma_{(N\!+\!1)} \hat e^\epsilon_{\mu;\emptyset} = b^\epsilon_\mu P^\gamma_{(N)}. \]
\end{thm}

\begin{proof}
We wish to prove $(P^\gamma_{(N\!+\!1)} \hat e^\epsilon_{\mu;\emptyset} f)(\bm x)= (b^\epsilon_\mu P^\gamma_{(N)} f)(\bm x)$ for all $f \in \ca C^\infty(\R^N)$ and all $\bm x \in \R^{N+1}$.
Because the operators involved preserve continuity, it is sufficient to check the statement for $\bm x$ in the dense set $\R^{N+1}_\text{reg}$.
Furthermore, because of \eqref{eqn:intertwine1} and $w^{\epsilon,\gamma} \hat e^\epsilon_{\mu;\emptyset} = \hat e^\epsilon_{\mu;\emptyset} w^\gamma$ for $w \in S_N$, we may assume $(x_1,\ldots,x_N) \in \R^N_+$, i.e. $\bm x \in \cup_{m=1}^{N\!+\!1} s_N \cdots s_m \R^{N\!+\!1}_+$ and it suffices to prove
\[ P^\gamma_{(N\!+\!1)} \hat e^\epsilon_{\mu;\emptyset} f|_{ s_N \cdots s_m \R^{N\!+\!1}_+} = 
b^\epsilon_\mu P^\gamma_{(N)} f|_{ s_N \cdots s_m \R^{N\!+\!1}_+}, \]
for all $m=1,\ldots,N+1$.
For $\epsilon=-$, this follows from \rfl{lem:eI4} and \rfl{lem:eI7}.
Using \rfl{lem:eplusemin} the results from Appendix \ref{sec:eI} can be straightforwardly modified to deal with $\epsilon=+$.
\end{proof}

\begin{thm}[Recursive construction for the non-symmetric wavefunction] \label{thm:psirecursion1}
Let $\bm \lambda \in \C^N$.
The recursions \eqref{eqn:psirecursion0} hold, and hence
\begin{equation} \label{eqn:psirecursion1} \psi_{\lambda_1,\ldots,\lambda_N} = b^-_{\lambda_N}  \cdots b^-_{\lambda_1} \Vac
= b^+_{\lambda_1}  \cdots b^+_{\lambda_N} \Vac \in \ca{CB}^\infty(\R^N).\end{equation}
\end{thm}

\begin{proof}
This follows from the fact that $\psi_{\bm \lambda}$ can be written in two ways using \rft{thm:bpropcommrel}:
\begin{align*}
\psi_{\bm \lambda} &= P^\gamma_{(N)}\e^{\ii \bm \lambda} &&= P^\gamma_{(N)} \hat e^-_{\lambda_N}  \e^{\ii (\lambda_1, \ldots,\lambda_{N\!-\!1})} 
&&= b^-_{\lambda_{N}} P^\gamma_{(N-1)} \e^{\ii (\lambda_1, \ldots,\lambda_{N\!-\!1})} 
&&= b^-_{\lambda_{N}} \psi_{\lambda_1, \ldots,\lambda_{N\!-\!1}}, \, \\
\psi_{\bm \lambda} &= P^\gamma_{(N)}\e^{\ii \bm \lambda} &&= P^\gamma_{(N)}\hat e^+_{\lambda_1}  \e^{\ii (\lambda_2, \ldots,\lambda_{N})} 
&&= b^+_{\lambda_1} P^\gamma_{(N-1)} \e^{\ii (\lambda_2, \ldots,\lambda_{N})} 
&&= b^+_{\lambda_{1}} \psi_{\lambda_2, \ldots,\lambda_{N}}.  \tag*{\qedhere}
\end{align*}
\end{proof}

We also recover the restrictions of \eqref{eqn:Dunklbcommrel2}-\eqref{eqn:Dunklbcommrel1} to the $P^\gamma_{(N)}$-image of $\ca{C}(\R^N)$ by using arguments such as
\[ \partial_j^{\epsilon,\gamma}  b^\epsilon_\mu P^\gamma_{(N)} =  \partial_j^{\epsilon,\gamma}  P^\gamma_{(N+1)} \hat e^\epsilon_{\mu;\emptyset} =  P^\gamma_{(N+1)} \partial^\epsilon_j \hat e^\epsilon_{\mu;\emptyset} =  P^\gamma_{(N+1)} \hat e^\epsilon_{\mu;\emptyset} \partial_j = b^\epsilon_\mu P^\gamma_{(N)} \partial_j = b^\epsilon_\mu \partial_j^\gamma P^\gamma_{(N)}. \]

\subsection{Commuting the $b^\pm_\mu$ on the span of the $\psi_{\bm \lambda}$} \label{sec:psispan}

Consider the following subspace of $\ca{CB}^\infty(\R^N)$:
\[ \f z_N := \Sp\set{\psi_{\bm \lambda}}{\bm \lambda \in \C^N}. \]
By virtue of the recursions \eqref{eqn:psirecursion0} both $b^\pm_\mu$ map $\f z_N$ to $\f z_{N+1}$, for all $\mu \in \C$.
We will now study the commutation relations among the $b^\pm_\mu$ acting on these subspaces $\f z_N$.
First of all, from \eqref{eqn:psirecursion0} it is immediately clear that
\begin{prop} 
For all $\mu,\nu \in \C$ we have
\begin{equation}
 [b^+_\mu, b^-_\nu] = 0 \in \Hom(\f z_N,\f z_{N\!+\!2}).  \label{eqn:bbcommrel1}
\end{equation}
\end{prop}

As for the commutation relation of two creation operators of the \emph{same} sign, we have
\begin{prop} 
Given $\epsilon = \pm$ and $\mu,\nu \in \C$ we have 
\begin{equation} \label{eqn:bbcommrel2}
s^\epsilon b^\epsilon_\mu b^\epsilon_\nu - b^\epsilon_\nu b^\epsilon_\mu = \frac{-\epsilon\ii \gamma}{\mu-\nu} \left[ b^\epsilon_\mu, b^\epsilon_\nu \right] \in \Hom(\f z_N,\f z_{N\!+\!2}).
\end{equation}
For $\mu=\nu$ the right-hand sides represent the limits $-\epsilon \ii \gamma \lim_{\nu \to \mu} \frac{[b^\epsilon_\mu, b^\epsilon_\nu]}{\mu-\nu}$.
\end{prop}

\begin{proof}
It suffices to prove the stated identities when applied to $\psi_{\bm \lambda}$ with $\bm \lambda \in \C^N$ arbitrary.
By virtue of \eqref{eqn:psirecursion0} in the case $\mu \ne \nu$ we need to show that
\begin{align*}
s_1 \psi_{\mu,\nu,\bm \lambda} &= \psi_{\nu,\mu,\bm \lambda} - \frac{\ii \gamma}{\mu-\nu} \left( \psi_{\mu,\nu,\bm \lambda}-\psi_{\nu,\mu,\bm \lambda}\right), \\
s_{N+1} \psi_{\bm \lambda,\nu,\mu} &= \psi_{\bm \lambda,\mu,\nu} + \frac{\ii \gamma}{\mu-\nu} \left( \psi_{\bm \lambda,\nu,\mu}-\psi_{\bm \lambda,\mu,\nu}\right),
\end{align*}
but this is precisely what is stated in \rfl{lem:Weylgponpsi} with $N \to N+2$, for $j=1$ (taking $\lambda_1 = \mu,\lambda_2 = \nu$) and $j=N+1$ (taking $\lambda_{N\!-\!1} = \nu, \lambda_{N+2} = \mu$).
The case $\mu = \nu$ follows from the case $\mu \ne \nu$ by taking limits and noting that the propagation operator preserves continuity.
\end{proof}

\begin{rem}
We do not claim that $\f z_N = \ca{CB}^\infty(\R^N)$.
Therefore we do not prove that the commutation relations among the $b^\pm_\mu$ hold on the entire $\ca{CB}^\infty(\R^N)$, although we conjecture this. 
\end{rem}

\section{The non-symmetric Yang-Baxter algebra} \label{sec:5}

There is a natural embedding of $\ca{C}(J^N)$ into the Hilbert space $\f h_N=\L^2(J^N)$, which contains the dense subspace
\begin{equation}
\f d_N   := \ca C^\infty_\text{cpt}(J^N)
\end{equation}
of \emph{test functions}, viz. smooth functions with compact support.
The \emph{(non-symmetric) Fock space} $\f h$, the subset of the direct sum of all $\f h_N$ consisting of elements of finite $\L^2$-norm,
is also a Hilbert space with respect to $\innerrnd{}{}$ and contains the dense subspace of \emph{finite vectors}
\begin{equation}
\f h_\text{fin} := \set{f \in \f h}{\exists M \in \Z_{\geq 0}: \, f \in \bigoplus_{N=0}^M \f h_N}.
\end{equation}
By restricting the arguments of functions acted upon by $\hat \phi_j, \hat e^\pm_{\mu;\bm i}, b^\pm_\mu$ to $J$ we may view these operators as acting on $\ca{C}(J^N)$. Furthermore, they are densely-defined linear maps$: \f h \to \f h$.\\

From now on, throughout sections \ref{sec:5} and \ref{sec:6}, assume that $J$ is bounded with $x^--x^+=L$.
We note that, given $\mu$, the elementary integral operators $\hat e^\pm_{\mu;\bm i}$, and hence the operators $b^\pm_\mu$, are bounded operators\footnote{A proof for this statement could go along the same lines as \cite[Props. 6.2.1~and~6.2.2]{Gutkin1988}.} and may therefore be considered as elements of $\End(\f h)$. This means they can be composed with other elements of $\End(\f h)$. 
We will now construct operators $a^\pm_\mu$ and $c^\pm_\mu$ out of the non-symmetric creation operators $b^\pm_\mu$ and show that they satisfy commutation relations akin to some of the relations in \eqref{eqn:ABCDcommrel}.

\subsection{The operators $a^\pm_\mu$}

The assignments
\[ (\check \phi^+ f)(\bm x) = f(x_+,\bm x), \qquad (\check \phi^- f)(\bm x) = f(\bm x,x_-) , \]
for $f \in \f h_{N+1}$, $\bm x \in J^N$ define $\check \phi^\pm \in \Hom(\f h_{N+1},\f h_N)$.
We have
\begin{gather}
[\check \phi^+,\check \phi^-]=0 , \label{eqn:checkphiproperty1} \\
(\check \phi^\epsilon )^2(1-s^\epsilon )=0 \qquad \text{for } \epsilon =\pm. \label{eqn:checkphiproperty2}
\end{gather}

\begin{lem}
Let $\epsilon = \pm$, $\mu \in \C$, $n=0,\ldots,N-1$ and $\bm i \in \f i^n_{[1,N\!-\!1]}$. Then in $\End(\f h)$ we have
\begin{align} 
\hat e^\epsilon_{\mu;\bm i} \check \phi^\epsilon &= \check \phi^\epsilon  s^\epsilon  \hat e^\epsilon_{\mu;\bm i^\epsilon }, \label{eqn:echeckphi} \\ 
b^\epsilon _\mu \check \phi^\epsilon & = \check \phi^\epsilon  s^\epsilon  b^\epsilon _\mu. \label{eqn:bcheckphi}
\end{align}
\end{lem}

\begin{proof}
\eqref{eqn:echeckphi} follows from the equalities
\begin{align*} 
(\hat e^+_{\mu;\bm i} \check \phi^+f)(x_1,\ldots,x_N) &= \e^{\ii \mu x_1} \theta_{\bm i^+ \, 1}(\bm x) 
\Biggl( \prod_{m=1}^n \int_{x_{i_{m+1}^+}}^{x_{i_m^+}} \dd y_m \e^{\ii \mu (x_{i_m^+}-y_m)} \Biggr) (\bar \phi_{\bm i^+}(\bm y) f)(x^+,x_2,\ldots,x_N) \\
&= (\hat e^+_{\mu;\bm i^+} f)(x_1,x^+,x_2,\ldots,x_N), \end{align*}
and
\begin{align*} 
(\hat e^-_{\mu;\bm i} \check \phi^-f)(x_1,\ldots,x_N) &=  \e^{\ii \mu x_N} \theta_{N \, \bm i}(\bm x) 
\Biggl( \prod_{m=1}^n \int_{x_{i_m}}^{x_{i_{m-1}}} \dd y_m \e^{\ii \mu (x_{i_m}-y_m)} \Biggr) (\bar \phi_{\bm i}(\bm y) f)(x_1,\ldots,x_{N\!-\!1},x^-) \\
&= (\hat e^-_{\mu;\bm i^-} f)(x_1,\ldots,x_{N\!-\!1},x^-,x_N), 
\end{align*}
respectively, for arbitrary $f \in \f h_N$ and $\bm x = (x_1,\ldots,x_N) \in J^N$.\\

To demonstrate \eqref{eqn:bcheckphi}, we show that $\sum_{\bm i \in \f i^n_{[1,N\!-\!1]}} \hat e^+_{\mu;\bm i} \check \phi^+ = \sum_{\bm i \in \f i^n_{[1,N]}} \check \phi^+ s^+ \hat e^+_{\mu;\bm i}$ for $x_1,\ldots,x_N>x^+$.
Indeed, 
\[ \sum_{\bm i \in \f i^n_{[1,N\!-\!1]}} \hat e^+_{\mu;\bm i} \check \phi^+ = \sum_{\bm i \in \f i^n_{[1,N\!-\!1]}} \check \phi^+ s^+ \hat e^+_{\mu;\bm i^+} = \! \!  \sum_{\bm i \in \f i^n_{[2,N]}} \! \! \check \phi^+ s^+ \hat e^+_{\mu;\bm i} = \sum_{\bm i \in \f i^n_{[1,N]}} \check \phi^+ s^+ \hat e^+_{\mu;\bm i},\]
where we have applied \eqref{eqn:echeckphi} and used that $\check \phi^+ s^+ \hat e^+_{\mu;\bm i}=0$ if $i_m=1$ for some $m$.
A similar argument can be made for the product $\check \phi^- s^- b^-_{\mu}$.
\end{proof}

Out of $\check \phi^\pm$ and the non-symmetric particle creation operators $b^\pm_\mu$ two new operators can be constructed that are endomorphisms of $\f h_N$; in particular, they preserve the particle number.

\begin{defn}
Let $\epsilon  = \pm$, $\mu \in \C$, $n=0,\ldots,N$ and $\bm i \in \f i^n_{[1,N]}$. Define
\[ \bar e^\epsilon_{\mu;\bm i} = \check \phi^\epsilon  \hat e^\epsilon_{\mu;\bm i} \in \End(\f h_N) . \]
In other words,
\begin{align*}
( \bar e^+_{\mu;\bm i} f)(\bm x) &=  \e^{\ii \mu x^+} \theta_{\bm i}(\bm x) \left( \prod_{m=1}^n  \int_{x_{i_{m+1}}}^{x_{i_m}} \dd y_m \e^{\ii \mu (x_{i_m}-y_m)} \right) (\bar \phi_{\bm i}(\bm y) f)(\bm x), \displaybreak[2] \\
( \bar e^-_{\mu; \bm i} f)(\bm x) &= \e^{\ii \mu x^-} \theta_{\bm i}(\bm x) \left( \prod_{m=1}^n \int_{x_{i_m}}^{x_{i_{m-1}}} \dd y_m \e^{\ii \mu (x_{i_m} -y_m)} \right) (\bar \phi_{\bm i}(\bm y) f)(\bm x),
\end{align*}
for $f \in \f h_N$ and $\bm x \in J^N$, where $x_{i_{n+1}} = x^+$ and $x_{i_0} = x^-$.
Furthermore, define
\[ a^\epsilon_\mu  = \check \phi^\epsilon  b^\epsilon_\mu =  \sum_{n \geq 0} \gamma^n \sum_{\bm i \in \f i^n_{[1,N]}} \bar e^\epsilon_{\mu;\bm i} \in \End(\f h_N). \]
\end{defn}

Similar to the situation for $b^\pm_\mu$, the operators $a^\pm_\mu$ are bounded on their domain of definition and may therefore be viewed as elements of $\End(\f h)$; in particular, they may be composed with other such elements.
From \eqref{eqn:bwcommrel}, for $w \in S_N$ and $\epsilon = \pm$, we obtain the identities
\begin{equation} \label{eqn:awcommrel}
[w, a^\epsilon_{\mu}] = 0 \in \End(\f h_N).
\end{equation}

\begin{lem} \label{lem:abrel1}
Let $\mu \in \C$ and $\epsilon =\pm$. We have $b^\epsilon_\mu a^\epsilon_{\nu} = \check \phi^\epsilon s^\epsilon b^\epsilon_\mu  b^\epsilon_\nu \in \Hom(\f h_N,\f h_{N+1})$. 
\end{lem}

\begin{proof}
Directly from \eqref{eqn:bcheckphi} and the definitions of $a^\pm_{\mu}$.
\end{proof}

\begin{lem} \label{lem:adjointnessa}
Let $\mu \in \C$, $n=0,\ldots,N$ and $\bm i \in \f i^n_{[1,N]}$, we have the formal adjointness relations
\[ \left( \bar e^+_{\mu;\bm i} \right)^* = \e^{-\ii \bar \mu(x^++x^-)} \bar e^-_{\bar \mu;\bm i}, \qquad 
\left( a^+_\mu \right)^* = \e^{-\ii \bar \mu(x^++x^-)} a^-_{\bar \mu}. \]
\end{lem}

\begin{proof}
The first statement can be easily checked by comparing the inner products $\innerrnd{\bar e^+_{\mu;\bm i} f}{g}$ and $\innerrnd{f}{\bar e^-_{\bar \mu;\bm i} g}$ for arbitrary $f,g \in \f d_N$. The statement relating $a^+_{\mu}$ and $a^-_{\bar \mu}$ is obtained by summing over all tuples $\bm i$.
\end{proof}

\subsection{The operators $c^\pm_\mu$}

The formal adjoints of $b^\pm_\mu$ have not been considered yet.

\begin{defn}
Let $\epsilon = \pm$, $\mu \in \C$, $n=0,\ldots,N$ and $\bm i \in \f i^n_{[1,N]}$, the operators $\check e^\epsilon_{\mu;\bm i} \in \Hom(\f h_{N+1}, \f h_N)$ are densely defined by
\begin{align*}
(\check e^+_{\mu;\bm i} f)(\bm x) &=  \e^{\ii \mu x^+} \theta_{\bm i}(\bm x) \left( \prod_{m=0}^n  \int_{x_{i_{m+1}}}^{x_{i_m}} \dd y_m \e^{\ii \mu (x_{i_m}-y_m)} \right) (\bar \phi_{\bm i}(\bm y) f)(\bm x,y_0), \\
(\check e^-_{\mu; \bm i} f)(\bm x) &= \e^{\ii \mu x^-} \theta_{\bm i}(\bm x) \left( \prod_{m=1}^{n+1} \int_{x_{i_m}}^{x_{i_{m-1}}} \dd y_m \e^{\ii \mu (x_{i_m} -y_m)} \right) (\bar \phi_{\bm i^+}(\bm y) f)(y_{n+1}, \bm x),
\end{align*}
for $f \in \f d_{N+1}$ and $\bm x \in J^N$. 
where again $x_{i_{n+1}} = x^+$ and $x_{i_0} = x^-$.
Furthermore we define $c^\epsilon_\mu \in \End(\f h)$ by $c^\epsilon_\mu = 0$ on $\f h_0$ and
\[ c^\epsilon_\mu|_{\f h_{N+1}} = \sum_{n=0}^N \gamma^{n+1} \sum_{\bm i \in \f i^n_{[1,N]}} \check e^\epsilon_{\mu;\bm i} \in \Hom(\f h_{N+1}, \f h_N). \]
\end{defn}

\begin{lem} \label{lem:adjointnessbc}
Given $\epsilon=\pm$, $\mu \in \C$, $n=0,\ldots,N$ and $\bm i \in \f i^n_{[1,N]}$, we have the formal adjointness relations
\[ \left( \hat e^\epsilon_{\mu;\bm i} \right)^* = \e^{-\ii \bar \mu(x^++x^-)} \check e^{-\epsilon}_{\bar \mu;\bm i}, \qquad
\left( b^\epsilon_\mu \right)^* = \gamma^{-1} \e^{-\ii \bar \mu(x^++x^-)} c^{-\epsilon}_{\bar \mu}. \]
\end{lem}

\begin{proof}
In the same way as for \rfl{lem:adjointnessa}.
\end{proof}

\begin{lem} 
Given $\epsilon = \pm$ and $\mu \in \C$ we have, in $\Hom(\f h_{N+1},\f h_N)$,
\begin{equation} 
c^\epsilon_{\mu} = [\check \phi^{-\epsilon}, a^\epsilon_\mu ] = \check \phi^\epsilon [\check \phi^{-\epsilon}, b^\epsilon_\mu ].  \label{eqn:cintermsofb} 
\end{equation}
\end{lem}

\begin{proof}
For $\bm x \in J^N$, in $(b^+_\mu f)(x^+,\bm x,x^-)$ split the summation over $\bm i \in \f i^n_{[1,N]}$ according to whether $i_1$ equals $N$ or not, corresponding to the terms appearing in $(c^+_{\mu}f)(\bm x)$ and $(a^+_\mu f)(x^+,\bm x)$, respectively. 
Then use $\check \phi^+ \check \phi^- b^+_\mu $ = $\check \phi^- a^+_\mu $.
A similar argument for $c^-_{\mu}$ is used.
\end{proof}

Again, it can be verified that the operators $c^\pm_\mu$ are bounded on their domain of definition and may therefore be viewed as elements of $\End(\f h)$; in particular, they may be composed with other such elements.
First of all, from \eqref{eqn:bwcommrel} and \rfl{lem:adjointnessbc} for $w \in S_N$ and $\epsilon = \pm$ we obtain
\begin{equation} \label{eqn:cwcommrel}
w c^\epsilon_\mu = c^\epsilon_\mu w^{-\epsilon} \in \Hom(\f h_{N+1},\f h_N).
\end{equation}

\subsection{Non-symmetric Yang-Baxter relations} \label{sec:nonsymmYBA}

The operators $a^\pm_\mu,b^\pm_\mu,c^\pm_\mu$ generate a subalgebra of $\End(\f h)$ which we will call the \emph{non-symmetric Yang-Baxter algebra}.
Given $\mu,\nu \in \C$ we can formulate commutation relations, to be referred to as the \emph{(non-symmetric) Yang-Baxter relations}, between $a^\pm_\mu$, $b^\pm_\mu$ and $c^\pm_\mu$, on the subspace $\bar{\f z}_N$, given by
\[ \bar{\f z}_N:=\overline{\Sp\set{\psi_{\bm \lambda}|_{J^N}}{\bm \lambda \in \C^N}} \subset \f h_N.\]

\begin{thm} \label{thm:abcommrel}
Given $\epsilon = \pm$, $\mu,\nu \in \C$, we have
\begin{align} 
\label{eqn:abcommrel} a^\epsilon_\mu b^\epsilon_\nu &= \frac{\mu-\nu+ \epsilon \ii \gamma}{\mu-\nu} b^\epsilon_\nu a^\epsilon_\mu -  \frac{\epsilon \ii \gamma}{\mu-\nu} b^\epsilon_\mu a^\epsilon_\nu &&\hspace{-20mm} \in \Hom(\bar{\f z}_N,\bar{\f z}_{N+1}), \\
\label{eqn:cacommrel} c^\epsilon_\mu a^\epsilon_\nu &= \frac{\mu-\nu - \epsilon \ii \gamma}{\mu-\nu} a^\epsilon_\nu c^\epsilon_\mu + \frac{\epsilon \ii \gamma}{\mu-\nu} a^\epsilon_\mu c^\epsilon_\nu &&\hspace{-20mm}\in \Hom(\bar{\f z}_{N+1},\bar{\f z}_N).
\end{align}
\end{thm}

\begin{proof}
Left-multiplying \eqref{eqn:bbcommrel2} by $\check \phi^\epsilon s^\epsilon$ yields
\[ \check \phi^\epsilon b^\epsilon_\mu b^\epsilon_\nu - \check \phi^\epsilon s^\epsilon b^\epsilon_\nu b^\epsilon_\mu = \frac{-\epsilon \ii \gamma}{\mu-\nu} \left( \check \phi^\epsilon s^\epsilon b^\epsilon_\mu b^\epsilon_\nu  - \check \phi^\epsilon s^\epsilon b^\epsilon_\nu b^\epsilon_\mu \right). \]
\rfl{lem:abrel1} and the definition of $a^\epsilon_\mu$ now imply \eqref{eqn:abcommrel}.
We obtain \eqref{eqn:cacommrel} by taking adjoints.
\end{proof}

\begin{thm}
Given $\epsilon = \pm$, $\mu,\nu \in \C$, we have
\begin{equation} \label{eqn:aacommrel}
\left[ a^\epsilon_\mu, a^\epsilon_\nu \right] = 0 \in \End(\bar{\f z}_N).
\end{equation}
\end{thm}

\begin{proof}
By virtue of \eqref{eqn:bcheckphi} and \eqref{eqn:checkphiproperty2} we have
\[ a^\epsilon_\mu a^\epsilon_\nu  = \phi^\epsilon b^\epsilon_\mu \phi^\epsilon b^\epsilon_\nu = (\phi^\epsilon)^2 s^\epsilon b^\epsilon_\mu b^\epsilon_\nu = (\phi^\epsilon)^2 b^\epsilon_\mu b^\epsilon_\nu. \]
Clearly it suffices to prove that
\begin{equation} \label{eqn:48}
(\phi^\epsilon)^2 b^\epsilon_\nu b^\epsilon_\mu=(\phi^\epsilon)^2 b^\epsilon_\mu b^\epsilon_\nu.
\end{equation}
\eqref{eqn:bbcommrel2} yields
\[ b^\epsilon_\nu b^\epsilon_\mu = \frac{(\mu-\nu)s^\epsilon + \epsilon \ii \gamma}{\mu-\nu + \epsilon \ii \gamma} b^\epsilon_\mu b^\epsilon_\nu; \]
left-multiplying by $(\phi^\epsilon)^2$ and applying \eqref{eqn:checkphiproperty2} again we obtain \eqref{eqn:48}.
\end{proof}

By restricting \eqref{eqn:bbcommrel1} to $J$ and taking its formal adjoint we obtain
\begin{thm} \label{thm:cc}
Given $\mu,\nu \in \C$ we have
\begin{gather}
\left[ b^+_\mu, b^-_\nu \right] = 0 \in  \Hom(\bar{\f z}_N,\bar{\f z}_{N+2}), \label{eqn:bbcommrel3}\\
\left[ c^+_\mu, c^-_\nu \right] = 0 \in \Hom(\bar{\f z}_{N+2},\bar{\f z}_{N}). \label{eqn:cccommrel}
\end{gather}
\end{thm}

We also obtain a commutation relation involving $a^+_\mu$ and $a^-_\nu$.
\begin{thm}
Given $\mu,\nu \in \C$ we have
\begin{equation} 
[a^+_\mu,a^-_\nu] = c^-_{\nu}b^+_\mu  - c^+_{\mu}b^-_{\nu} \in \End(\bar{\f z}_N); \label{eqn:aacommrel2} 
\end{equation}
In particular, $[a^+_\mu ,a^-_{\nu}]$ is not invariant under $\mu \leftrightarrow \nu$.
\end{thm}

\begin{proof}
Focusing on the right-hand side, we have
\[  c^-_{\nu}b^+_\mu  - c^+_{\mu}b^-_{\nu}  = \check \phi^+ a^-_{\nu} b^+_\mu  - a^-_{\nu} \check \phi^+ b^+_\mu  - \check \phi^- a^+_\mu  b^-_{\nu} + a^+_\mu  \check \phi^- b^-_{\nu}, \]
by virtue of \eqref{eqn:cintermsofb}. Now using the definition of $a^\pm_\mu$ in terms of $b^\pm_\mu$ as well as \eqref{eqn:bbcommrel3} and \eqref{eqn:checkphiproperty1} we obtain the result.
\end{proof}

\subsection{A partial algebraic Bethe ansatz for the $\psi_{\bm \lambda}$}

The statements \rft{thm:abcommrel} allow us to express $a^\pm_\mu \psi_{\bm \lambda}$ as linear combinations of eigenfunctions $\psi_{\bm \nu}$ where $\bm \nu$ ranges over the set of $N$-tuples whose entries are distinct elements of the set $\{ \mu,\lambda_1,\ldots,\lambda_N\}$.

\begin{thm} \label{thm:apsi}
Given $n=0,\ldots,N$, $\bm i \in \f I^n_{[1,N]}$, $\bm \lambda \in \C^N$ and $\lambda_0,\lambda_{N+1} \in \C$,
we write $i_0 = 0$ and $i_{n+1} = N+1$ and we have
\begin{align}
a^+_{\lambda_0} \psi_{\bm \lambda} &=\sum_{n=0}^N \sum_{\bm i \in \f I^n_{[1,N]}} a^+_{\bm i}(\bm \lambda;{\lambda_0}) \psi_{\lambda_1,\ldots,\underset{(i_1)}{\lambda_{i_0}},\ldots,\underset{(i_2)}{\lambda_{i_1}},\ldots,\underset{(i_n)}{\lambda_{i_{n-1}}},\ldots,\lambda_N}, \label{eqn:apluspsi}  \\
a^-_{\lambda_{N+1}} \psi_{\bm \lambda} &= \sum_{n=0}^N \sum_{\bm i \in \f I^n_{[1,N]}} a^-_{\bm i}(\bm \lambda;\lambda_{N+1}) \psi_{\lambda_1,\ldots,\underset{(i_1)}{\lambda_{i_2}},\ldots,\underset{(i_2)}{\lambda_{i_3}},\ldots,\underset{(i_n)}{\lambda_{i_{n+1}}},\ldots,\lambda_N}. \label{eqn:aminpsi}
\end{align}
The complex numbers $a^\pm_{\bm i}(\bm \lambda;\mu)$ are defined by the recursions
\begin{align*}
a^+_{\bm i^+}(\lambda_1,\ldots,\lambda_N;\mu) &= \frac{\lambda_1\!-\!\mu\!-\!\ii\gamma}{\lambda_1\!-\!\mu} a^+_{\bm i}( \lambda_2,\ldots,\lambda_N;\mu), && \text{for } \bm i \in \f I^n_{[1,N)}, \\
a^+_{1 \, (\bm i')^+}(\lambda_1,\ldots,\lambda_N;\mu) &= \frac{\ii \gamma}{\lambda_1\!-\!\mu} a^+_{\bm i'}(\lambda_2,\ldots,\lambda_N;\lambda_1), && \text{for } \bm i' \in \f I^{n-1}_{[1,N)}, \\
a^-_{\bm i}(\lambda_1,\ldots,\lambda_N;\mu) &= \frac{\lambda_{N}\!-\!\mu\!+\!\ii\gamma}{\lambda_{N}\!-\!\mu} a^-_{\bm i}( \lambda_1,\ldots,\lambda_{N-1};\mu), && \text{for } \bm i \in \f I^n_{[1,N)}, \\
a^-_{\bm i' \, N}(\lambda_1,\ldots,\lambda_N;\mu) &= \frac{\!-\!\ii\gamma}{\lambda_{N}\!-\!\mu} a^-_{\bm i'}(\lambda_1,\ldots,\lambda_{N-1};\lambda_N), && \text{for } \bm i' \in \f I^{n-1}_{[1,N)},
\end{align*}
and the initial values $a^\epsilon_{\emptyset}(\emptyset;\mu) = \e^{\ii \mu x^\epsilon}$ for $\epsilon = \pm$.
\end{thm}

\begin{proof}
We present the proof for the expression for $a^-_{\lambda_{N+1}} \psi_{\bm \lambda}$; the expression for $a^+_{\lambda_0} \psi_{\bm \lambda}$ is established along the same lines.
The proof is by induction on $N$; the $N=0$ case reproduces $a^-_{\lambda_{N+1}} \Vac = \e^{\ii \lambda_{N+1} x^-} \Vac$.
Assuming the statement for $N$, we will prove it with $N$ replaced by $N+1$, using \eqref{eqn:psirecursion0} and \eqref{thm:abcommrel}.
Writing $\bm \lambda = (\lambda_1,\ldots,\lambda_N)$ we have
\[ a^-_{\lambda_{N\!+\!2}}\psi_{\bm \lambda,\lambda_{N\!+\!1}} 
= \frac{\lambda_{N\!+\!1}\!-\!{\lambda_{N\!+\!2}}\!+\!\ii \gamma}{\lambda_{N\!+\!1}\!-\!{\lambda_{N\!+\!2}}} b^-_{\lambda_{N\!+\!1}} a^-_{\lambda_{N\!+\!2}} \psi_{\bm \lambda}- \frac{\ii \gamma}{\lambda_{N\!+\!1}\!-{\lambda_{N\!+\!2}}} b^-_{\lambda_{N\!+\!2}} a^-_{\lambda_{N\!+\!1}} \psi_{\bm \lambda}. \]
Using the induction hypothesis, we have
\begin{align*} 
a^-_{{\lambda_{N\!+\!2}}}\psi_{\bm \lambda,\lambda_{N\!+\!1}} 
&= \sum_{n=0}^N \sum_{\bm i \in \f I^n_{[1,N]}} \frac{\lambda_{N\!+\!1}\!-\!{\lambda_{N\!+\!2}}\!+\!\ii \gamma}{\lambda_{N\!+\!1}\!-\!{\lambda_{N\!+\!2}}} a^-_{\bm i}(\bm \lambda;{\lambda_{N\!+\!2}}) \e^{\ii \lambda_{i_1}x^-}  \psi_{\lambda_1,\ldots,\underset{(i_1)}{\lambda_{i_2}},\ldots,\underset{(i_n)}{{\lambda_{N\!+\!2}}},\ldots,\lambda_N,\lambda_{N\!+\!1}} + \\
& \quad + \sum_{n=0}^N \sum_{\bm i \in \f I^n_{[1,N]}} \frac{-\ii \gamma}{\lambda_{N\!+\!1}\!-\!{\lambda_{N\!+\!2}}} a^-_{\bm i}(\bm \lambda;\lambda_{N\!+\!1}) \e^{\ii \lambda_{i_1}x^-}  \psi_{\lambda_1,\ldots,\underset{(i_1)}{\lambda_{i_2}},\ldots,\underset{(i_n)}{\lambda_{N\!+\!1}},\ldots,\lambda_N,{\lambda_{N\!+\!2}}} \displaybreak[2] \\
&= \sum_{n=0}^N \sum_{\bm i \in \f I^n_{[1,N]}} a^-_{\bm i}(\bm \lambda,\lambda_{N\!+\!1};{\lambda_{N\!+\!2}}) \e^{\ii \lambda_{i_1}x^-}  \psi_{\lambda_1,\ldots,\underset{(i_1)}{\lambda_{i_2}},\ldots,\underset{(i_n)}{\lambda_{i_{n+1}}},\ldots,\lambda_N,\lambda_{N\!+\!1}} + \\
& \quad + \sum_{n=1}^{N\!+\!1} \sum_{\atop{\bm i \in \f I^n_{[1,N\!+\!1]}}{ i_n={N\!+\!1}}}  a^-_{\bm i}(\bm \lambda,\lambda_{N\!+\!1};{\lambda_{N\!+\!2}})  \e^{\ii \lambda_{i_1}x^-}  \psi_{\lambda_1,\ldots,\underset{(i_1)}{\lambda_{i_2}},\ldots,\underset{(i_{n-1})}{\lambda_{i_n}},\ldots,\lambda_N,\underset{(N\!+\!1)}{{\lambda_{N\!+\!2}}}} 
\end{align*}
where we have used the recursion for the coefficient functions $a^-_{\bm i}(\bm \lambda;\mu)$.
Hence, using the decomposition $\f I^n_{[1,N\!+\!1]} = \left( \f I^{n-1}_{[1,N]} \times \{ N+1 \} \right) \cup \f I^n_{[1,N]}$ into disjoint subsets, we indeed obtain,
\begin{align*}
a^-_{{\lambda_{N+2}}}\psi_{\bm \lambda,\lambda_{N+1}} &= \sum_{n=0}^{N\!+\!1} \left( \sum_{\bm i \in \f I^n_{[1,N]}} a^-_{\bm i}(\bm \lambda,\lambda_{N+1};{\lambda_{N+2}}) \e^{\ii \lambda_{i_1}x^-}  \psi_{\lambda_1,\ldots,\underset{(i_1)}{\lambda_{i_2}},\ldots,\underset{(i_n)}{\lambda_{i_{n+1}}},\ldots,\lambda_N,\lambda_{N+1}} + \right. \\
& \qquad \qquad 
\left. +
\sum_\atop{\bm i \in \f I^n_{[1,N\!+\!1]} }{ i_n={N\!+\!1}}  a^-_{\bm i}(\bm \lambda,\lambda_{N+1};{\lambda_{N+2}})  \e^{\ii \lambda_{i_1}x^-}  \psi_{\lambda_1,\ldots,\underset{(i_1)}{\lambda_{i_2}},\ldots,\underset{(i_{n-1})}{\lambda_{i_n}},\ldots,\lambda_N,\underset{(N+1)}{{\lambda_{N+2}}}} 
\right) \\
&= \sum_{n=0}^{N\!+\!1} \sum_{\bm i \in \f I^n_{[1,N\!+\!1]}} a^-_{\bm i}(\bm \lambda,\lambda_{N+1};{\lambda_{N+2}}) \e^{\ii \lambda_{i_1}x^-}  \psi_{\lambda_1,\ldots,\underset{(i_1)}{\lambda_{i_2}},\ldots,\underset{(i_n)}{{\lambda_{i_{n+1}}}},\ldots,\lambda_{N+1}}.  \tag*{\qedhere}
\end{align*}
\end{proof}

\begin{rem}
\rft{thm:apsi} provides a partial analogue of the ABA in the non-symmetric setting. In the symmetric case, the action of the operators $A_\mu, D_\mu$ on $\Psi_{\bm \lambda} = \ca S_{(N)} \psi_{\bm \lambda}$ can be obtained by replacing $a^+_\mu \to A_\mu, a^-_\mu \to D_\mu$ in \eqref{eqn:apluspsi}-\eqref{eqn:aminpsi}; then identities such as
\begin{equation}  \frac{-\ii \gamma}{\lambda_j \! - \! \mu} \frac{\lambda_k\!-\!\mu\!+\!\ii \gamma}{\lambda_k\!-\!\mu} + \frac{- \ii \gamma}{\lambda_j \! - \! \lambda_k}  \frac{-\ii \gamma}{\lambda_k\!-\!\mu} = \frac{-\ii \gamma}{\lambda_j \! - \! \mu} \frac{\lambda_j \! - \! \lambda_k \! - \! \ii \gamma}{\lambda_j \! - \! \lambda_k } \label{eqn:coefficientidentity} \end{equation}
allow one to combine coefficients of the (equal) eigenfunctions $\Psi_{\ldots,\underset{j}{\mu},\ldots,\underset{k}{\lambda_k},\ldots}$ and $\Psi_{\ldots,\underset{j}{\lambda_k},\ldots,\underset{k}{\mu},\ldots}$ in the expansion of $D_\mu \Psi_{\bm \lambda}$. After having combined terms in the expansion of $A_\mu \Psi_{\bm \lambda}$ in an analogous manner one may set coefficients of ``unwanted'' terms in the expansion of $(A_\mu + D_\mu)\Psi_{\bm \lambda}$ to zero by imposing the BAEs \eqref{eqn:BAEintro}.
The failure of the ABA in the non-symmetric case owes to the fact that $\psi_{\ldots,\underset{(j)}{\lambda_j},\ldots,\underset{(k)}{\lambda_k},\ldots} \ne \psi_{\ldots,\underset{(j)}{\lambda_k},\ldots,\underset{(k)}{\lambda_j},\ldots}$ so that \eqref{eqn:coefficientidentity} cannot be used to combine coefficients.
\end{rem}

\section{Recovering the symmetric Yang-Baxter algebra} \label{sec:6}

Note that the commutation relations involving $a^\pm_\mu ,b^\pm_\mu,c^\pm_\mu$, viz. \eqref{eqn:abcommrel}, \eqref{eqn:cacommrel}, \eqref{eqn:aacommrel}, \eqref{eqn:bbcommrel3}, and \eqref{eqn:cccommrel} are of the exact same form as some of the established commutation relations involving $A_\mu,B_\mu,C_\mu,D_\mu$ appearing in \eqref{eqn:ABCDcommrel}, i.e. by replacing lower case letters by uppercase letters these non-symmetric and symmetric Yang-Baxter relations transform into each other. In fact, we can prove these relations, and further highlight why the operators $a^\pm_\mu ,b^\pm_\mu,c^\pm_\mu$ are relevant to the study of the QNLS model, by defining the operators $A_\mu,B_\mu,C_\mu,D_\mu$ in terms of the operators $a^\pm_\mu,b^\pm_\mu,c^\pm_\mu$.\\

Consider
\[ J^N_+ := \R^N_+ \cap (x^+,x^-)^N = \set{\bm x \in \R^N}{x^- > x_1>\ldots>x_N > x^+} \]
so that $J^N = \overline{\cup_{w \in S_N} w J^N_+}$.
Given $F \in \ca H_N$, $\mu \in \C$, $\bm i \in \f i^n_{[1,N]}$ and $\epsilon = \pm$, it can be checked that 
\[ \ca S_{(N\!+\!1)}\hat e^\epsilon_{\mu;\bm i} F|_{J^{N\!+\!1}_+} = \frac{(N-n)!}{(N+1)!} \sum_{\bm j \in \f I^{n+1}_{[1,N\!+\!1]}} \hat E_{\mu;\bm j} F|_{J^{N\!+\!1}_+}, \]
where $\hat E_{\mu;\bm i} \in \Hom(\ca H_N, \ca H_{N+1})$ is defined by
\begin{align*} 
\left(\hat E_{\mu;\bm i}F\right)(x_1,\ldots,x_{N+1}) &= \left( \prod_{m=1}^n \int^{x_{i_m}}_{x_{i_{m+1}}} \dd y_m \right) \e^{\ii \left( \sum_{m=1}^{n+1} x_{i_m} - \sum_{m=1}^n y_m \right)} \cdot \\
& \qquad F(x_1,\ldots,\widehat{x_{i_1}},\ldots,\widehat{x_{i_{n+1}}},\ldots,x_{N+1},y_1,\ldots,y_n) 
\end{align*} 
for $F \in \ca H_N$, $\bm x \in J^{N\!+\!1}_+$.
In particular, the restricted function $\ca S_{(N\!+\!1)}\hat e^\epsilon_{\mu;\bm i} F|_{J^{N\!+\!1}_+}$ is independent of $\bm i$ and $\epsilon$.
Since the cardinality of $\f i^n_{[1,N]}$ equals $\frac{N!}{(N-n)!}$, it follows that
\[ \ca S_{(N\!+\!1)} b^\pm_\mu F|_{J^{N\!+\!1}_+}= \frac{1}{N+1} \sum_{n=0}^{N} \gamma^n \sum_{\bm i \in \f I^{n+1}_{[1,N\!+\!1]}} \hat E_{\mu;\bm i} F|_{J^{N\!+\!1}_+}.\]
Hence we define
\begin{equation} \label{eqn:Boperator}
B_\mu := \ca S_{(N\!+\!1)} b^\pm_\mu|_{\ca H_N}  \in \Hom(\ca H_N, \ca H_{N+1}).
\end{equation}
This $B_\mu = \frac{1}{N+1} \sum_{n=0}^{N} \gamma^n \sum_{\bm i \in \f I^{n+1}_{[1,N\!+\!1]}} \hat E_{\mu;\bm i}$ is the known top-left entry of the QNLS monodromy matrix $\ca T_\mu$ \cite{Gutkin1988,KorepinBI,Sklyanin1982}.\\

Using \eqref{eqn:awcommrel} it is clear that $a^\pm_\mu$ maps $S_N$-invariant functions to $S_N$-invariant functions.
Equally from \eqref{eqn:cwcommrel} it follows that $c^\pm_\mu$ maps $S_{N+1}$-invariant functions to $S_N$-invariant functions; furthermore given $\mu \in \C$, $n=0,\ldots,N-1$ and $\bm i \in \f i^n_{[1,N\!-\!1]}$, the elementary operators $\check e^\pm_{\mu;\bm i}$ coincide on $\ca H_N$ as can be easily checked, so that $c^\pm_\mu$ coincide on $\ca H_N$.
Hence, we define the other QNLS monodromy matrix entries $A_\mu, C_\mu, D_\mu$ \cite{Gutkin1988,KorepinBI,Sklyanin1982} as follows
\begin{equation} \label{eqn:ACDoperators} A_\mu := a^+_\mu|_{\ca H_N} ,\,  D_\mu := a^-_\mu|_{\ca H_N}  \in \End(\ca H_N), \qquad
C_\mu := c^\pm_\mu|_{\ca H_{N+1}} \in \Hom(\ca H_{N+1},\ca H_N). \end{equation}
Now all commutation relations in \eqref{eqn:ABCDcommrel} except the ones involving all of $A_\mu,B_\mu,C_\mu,D_\mu$ can be derived.
For example, 
\begin{align*} 
\left[B_\mu,A_\nu \right] + \frac{\ii \gamma}{\mu \! - \! \nu} \left( A_\mu B_\nu  \!- \! A_\nu B_\mu \right) 
&= \left( B_\mu  a^+_\nu \! - \! A_\nu \ca S_{(N\!+\!1)} b^+_\mu + \frac{ \ii \gamma}{\mu \! - \! \nu} \left( A_\mu \ca S_{(N\!+\!1)} b^+_\nu \! - \! A_\nu \ca S_{(N\!+\!1)} b^+_\mu \right) \right)|_{\ca H_N} \\
&= \ca S_{(N\!+\!1)} \left( b^+_\mu a^+_\nu \! - \!  a^+_\nu b^+_\mu + \frac{ \ii \gamma}{\mu \! - \! \nu}  \left( a^+_\mu b^+_\nu  \! - \! a^+_\nu b^+_\mu \right) \right)|_{\ca H_N},
\end{align*}
which vanishes by virtue of \eqref{eqn:abcommrel}.

\begin{rem}
It remains an open problem to derive relations such as
\[ \left[A_\mu,D_\nu \right] = \frac{-\ii \gamma}{\mu - \nu} \left( B_\mu C_\nu  -  B_\nu C_\mu \right) \]
using this formalism. Equally, the non-symmetric commutation relations \eqref{eqn:aacommrel2} cannot be restricted to relations involving only the symmetric operators $A_\mu ,B_\mu , C_\mu, D_\mu$. 
%Moreover, the precise relation of the non-symmetric Yang-Baxter algebra to the YBE \eqref{eqn:YBE} and the Yangian of $\f{gl}_2$ is unclear.
We reiterate that the non-symmetric relations obtained in sections \ref{sec:psispan} and \ref{sec:nonsymmYBA} are proven on the closure of the span of the $\psi_{\bm \lambda}$ only. 
The symmetrized relations obtained here must correspondingly hold on the closure of the span of the $\Psi_{\bm \lambda}$, which is known \cite{Dorlas} to be equal to the whole $\ca H_N$.
\end{rem}

We may view the definition \eqref{eqn:Boperator} of $B_\mu$ as an identity in $\Hom(\ca{CB}^\infty(\R^N)^{S_N},\ca{CB}^\infty(\R^{N\!+\!1})^{S_{N+1}})$. 
Then by virtue of \eqref{eqn:bwcommrel} we obtain 
\begin{thm} \label{thm:bsymmcommrel}
Let $\mu \in \C$ and $\epsilon = \pm$. We have
\[ \ca S_{(N+1)} b^\epsilon_\mu  = B_\mu \ca S_{(N)} \in \Hom(\ca{CB}^\infty(\R^N),\ca{CB}^\infty(\R^{N\!+\!1})^{S_{N\!+\!1}}).\]
\end{thm}
Hence, we obtain a novel proof for the QISM recursion \eqref{eqn:Psirecursion}; writing $\bm \lambda = (\lambda_1,\ldots,\lambda_{N\!+\!1})$ and $\bm \lambda'= (\lambda_1,\ldots,\lambda_N)$ we have
\[ \Psi_{\bm \lambda} = \ca S_{(N\!+\!1)} \psi_{\bm \lambda} = \ca S_{(N\!+\!1)} b^-_{\lambda_{N+1}} \psi_{\bm \lambda'} = B_{\lambda_{N+1}} \ca S_{(N)} \psi_{\bm \lambda}  = B_{\lambda_{N+1}} \Psi_{\bm \lambda}.\]
Combining \rft{thm:bpropcommrel} and \rft{thm:bsymmcommrel}, we obtain the following scheme for the recursive construction of $\Psi_{\bm \lambda}$:
\[ \xymatrix@R=1.2cm@C=1.5cm{
1 \ar[d]^{P^\gamma_{(0)}} \ar[r]^{\hat e^-_{\lambda_1}} & \e^{\ii \lambda_1} \ar[d]^{P^\gamma_{(1)}} \ar[r]^{\hat e^-_{\lambda_2}} &  \ldots \ar[r]^{\hspace{-4mm} \hat e^-_{\lambda_{N\!-\!1}}} & \e^{\ii (\lambda_1,\ldots,\lambda_{N\!-\!1})} \ar[d]^{P^\gamma_{(N-1)}} \ar[r]^{\hat e^-_{\lambda_{N}}} & \e^{\ii (\lambda_1,\ldots,\lambda_{N})} \ar[d]^{P^\gamma_{(N)}}\\ 
1 \ar[d]^{\ca S_{(0)}} \ar[r]^{b^-_{\lambda_1}} & \psi_{\lambda_1} \ar[d]^{\ca S_{(1)}} \ar[r]^{b^-_{\lambda_2}} & \ldots \ar[r]^{\hspace{-4mm} b^-_{\lambda_{N\!-\!1}}} & \psi_{\lambda_1,\ldots,\lambda_{N\!-\!1}} \ar[d]^{\ca S_{(N\!-\!1)}} \ar[r]^{b^-_{\lambda_N}} & \psi_{\lambda_1,\ldots,\lambda_{N}}  \ar[d]^{\ca S_{(N)}} \\  
1 \ar[r]^{B_{\lambda_1}} & \Psi_{\lambda_1} \ar[r]^{B_{\lambda_2}} & \ldots \ar[r]^{\hspace{-4mm} B_{\lambda_{N\!-\!1}}} & \Psi_{\lambda_1,\ldots,\lambda_{N\!-\!1}} \ar[r]^{B_{\lambda_N}} & \Psi_{\lambda_1,\ldots,\lambda_{N}} } \]
Note that the three operators $\hat e^-_{\mu;\emptyset}$, $b^-_{\mu}$, $B_\mu $ coincide when acting on $\f h_0$ or $\f h_1$; equivalently, $P^\gamma_{(0)} = P^\gamma_{(1)} = \ca S_{(0)} = \ca S_{(1)} =1$. 
A second scheme may be created by replacing each operator acting horizontally in the above scheme by its $+$-version, and simultaneously each $\lambda_j$ by $\lambda_{N-j+1}$.

\section{Summary and conclusions} \label{sec:7}

The main purpose of this paper has been to tie closer together the two mathematical theories underlying the QNLS model (the dAHA method and the QISM) which we have theoretically motivated as resulting from a ``deformation'' of Schur-Weyl duality. 
In particular, in Section \ref{sec:4} we have described a vertex-type operator formalism generating the non-symmetric wavefunctions which have hitherto only been defined in terms of representations of the dAHA (Section \ref{sec:3}); put in different words, we have highlighted the (hidden) recursive structure in the dAHA method. 
A key result in this respect is Thm. \ref{thm:bpropcommrel}, the commutation relation of Gutkin's propagation operator (defined in terms of representations of the dAHA) with a particle creation operator, realized as a non-symmetric version of the vertex-type operators featuring in the QISM.
In Section \ref{sec:5} we have defined further non-symmetric vertex-type operators satisfying QISM-type commutation relations, and we have seen how the ordinary QISM can be recovered by symmetrizing the vertex operators introduced in this paper (Section \ref{sec:6}). \\

Further theoretical work directly related to the contents of this paper remains to be done; we highlight the following unresolved issues. 
\begin{enumerate}
\item[1.] It is not clear whether the set of non-symmetric wavefunctions $\set{ \psi_{\bm \lambda} }{ \bm \lambda \in \C^N}$ spans $\ca{CB}^\infty(\R^N)$.
It is tempting to pursue a proof of this by making use of a density property of the set of plane waves $\e^{\ii \bm \lambda}$ in a suitable function space, and inverting the propagation operator $P^\gamma_{(N)}$, but it unclear how to make this rigorous. 
The only invertibility result known to the author restricts the domain of $P^\gamma_{(N)}$ to analytic functions \cite[Thm. 5.3(ii)]{EmsizOS}.
On a related point, it might be feasible to prove the key commutation relations \eqref{eqn:bbcommrel1}-\eqref{eqn:bbcommrel2} on the entire $\ca{CB}^\infty(\R^N)$ by using the $\gamma$-expansions for $b^\pm_\mu$, thereby bypassing the need to restrict to the span of the $\psi_{\bm \lambda}$.
\item[2.] The precise relation of the operators $a^\pm_\mu,b^\pm_\mu,c^\pm_\mu$ to the Yangian $Y(\f{gl}_2)$ is unclear. 
In particular, it is not straightforward to ``combine''  these operators into non-symmetric analogons of the monodromy matrix $\ca T_\mu$ which should satisfy a variant of the Yang-Baxter equation \eqref{eqn:YBE} encoding the relations \eqref{eqn:abcommrel}, \eqref{eqn:cacommrel}, \eqref{eqn:aacommrel}, \eqref{eqn:bbcommrel3}, and \eqref{eqn:cccommrel}. 
Also, we reiterate that the commutation relations involving all four of the symmetrized operators $A_\mu$, $B_\mu$, $C_\mu$, $D_\mu$ do not have an analogon for the non-symmetric counterparts, and that relation \eqref{eqn:aacommrel2} does not immediately lead to a relation for symmetrized operators.
\end{enumerate}

The techniques discussed in this paper may be relevant or applicable to other physical models involving $\delta$-interaction. 
\begin{enumerate}
\item There exists a large body of theory dealing with the fermionic QNLS model (e.g. \cite{MurakamiWadati,MurakamiWadati2,Yang1967}). 
Without going in detail, we remark that the physically relevant wavefunctions can in principle be constructed from the same non-symmetric wavefunctions discussed in this article by antisymmetrization, i.e. by calculating $\frac{1}{N!} \sum_{w \in S_N} \sgn(w) \psi(x_{w 1},\ldots,x_{w N})$ and it should be possible to express fermionic particle creation operators in terms of the non-symmetric operators $b^\pm_\mu$ in analogy to \eqref{eqn:Boperator}. 
Again, this ties in with the idea that the non-symmetric wavefunctions are somehow the more fundamental objects in these theories, together with, we argue, the non-symmetric creation operators $b^\pm_\mu$.
More generally, this method may work for anyonic models with $\delta$-interaction \cite{Kundu1999,Kundu2010} and may even lead to a description of ``particle creation'' or ``raising'' operators for invariant wavefunctions. 
It is expected that in this case the dAHA needs to be replaced by some deformation of the group algebra of the \emph{braid group}, instead of the symmetric group.
Related to this, the compatibility of the vertex-type operator formalism discussed here with the notion of spin needs to be investigated in more detail, requiring vector analogues of the propagation operator \cite{Emsiz} and the non-symmetric creation operators. 
\item It would be interesting to see in how far the results of this paper can be generalized to other root systems, i.e. to systems with open or reflecting boundary conditions.
Quantum Hamiltonians can be defined in general terms, as can the Dunkl-type operator and propagation operator formalisms \cite{Emsiz,EmsizOS,Gaudin1971-3,Gutkin1982,GutkinSutherland,HeckmanOpdam1997,Hikami,KomoriHikami}. 
The problems lie on the side of the QISM; notwithstanding the discussion of the C-type formalism in \cite{Sklyanin1988}, 
there seems to be no general way to formulate monodromy matrices, or even particle creation operators, for an arbitrary root system.
However, mindful of the explicit connection between the dAHA approach and the QISM set out in this paper, one may be able to find particle creation operators for \emph{non-invariant wavefunctions} (with respect to the Weyl group associated to the root system) by discovering a ``natural'' recursion involving the propagation operators associated to root systems of the same type and subsequent rank.
It is key to assess in how far ``special'' properties of the A-type system (for example the realization in terms of the centre-of-mass frame) are crucial in the formulation of the particle creation operators.
Failing that, it would be interesting to see if the non-symmetric wavefunctions of A-type can be used directly to construct invariant wavefunctions of general type.
\end{enumerate}

Finally, we remark that it would be beneficial to investigate correlation functions, orthogonality relations and norms (with respect to the L$^2$-inner product) of the non-symmetric wavefunctions, although the required commutation relations between $c^\pm_\lambda$ and $b^\pm_\mu$ are as yet elusive.
This could be especially useful for systems that are not of A-type, since for these systems the corresponding norm formulae for $\Psi_{\bm \lambda}$ have been conjectured but not proven \cite{BustamanteVDDlM,Emsiz,Korepin}.

\appendix
\appendixpage

\section{Explicit formulae for $N=2$} \label{sec:example}

The Dunkl-type operators $\partial_1^\gamma $, $\partial_2^\gamma  \in \End(\ca C^\infty(\R^2_\text{reg}))$ are defined by
\begin{align*}
(\partial_1^\gamma  f)(x_1,x_2) &= (\partial_1 f)(x_1,x_2) +  \begin{cases} 0, & x_1>x_2 \\  \gamma f(x_2,x_1), & x_2>x_1, \end{cases} \\
(\partial_2^\gamma  f)(x_1,x_2) &= (\partial_2 f)(x_1,x_2) - \begin{cases} 0, & x_1>x_2 \\  \gamma f(x_2,x_1), & x_2>x_1, \end{cases}.
\end{align*}
$\partial_1^\gamma $, $\partial_2^\gamma$ and $s_1 \in \End(\ca C^\infty(\R^2_\text{reg}))$ defined by $(s_1 f)(x_1,x_2)=f(x_2,x_1)$ represent the dAHA $\ca{A}_2^\gamma$:
\[ s_1^2 = 1, \quad s_1 \partial_1^\gamma  - \partial_2^\gamma  s_1 = \gamma, \quad \left[\partial_1^\gamma ,\partial_2^\gamma \right]=0. \]
Furthermore, we have the integral-reflection operator $s_{1}^\gamma  = s_1 + \gamma I_{12} \in \End(\ca C^\infty(\R^2))$ where
\[(I_{12} f)(x_1,x_2) = \int_0^{x_1-x_2} \! \!  \dd y f(x_1-y,x_2+y) \]
which with the partial differential operators $\partial_1$, $\partial_2$ also represent $\ca{A}_2^\gamma$:
\[ (s_1^\gamma)^2 = 1, \quad s_{1}^\gamma  \partial_1 - \partial_2 s_{1}^\gamma  = \gamma, \quad \left[\partial_1,\partial_2 \right]=0. \]
Write $\theta_{12}$ for the multiplication operator corresponding to the characteristic function of the alcove $\R^2_+$; then $\theta_{21} = s_{12}\theta_{12}s_{12}$ corresponds to the characteristic function of the other alcove $s_1 \R^2_+$. The propagation operator $P^\gamma_{(2)} \in \Hom(\ca C^\infty(\R^2),\ca{CB}^\infty(\R^2))$ is defined by 
\[ P^\gamma_{(2)} = \begin{cases} 1, & \text{on } \R^2_+ \\ s_1 s_{1}^\gamma , & \text{on } s_1 \R^2_+ \end{cases}  = 1+\gamma\theta_{21} I_{21}. \]
$P^\gamma_{(2)}$ satisfies the intertwining relations
\[ s_1 P^\gamma_{(2)} = P^\gamma_{(2)} s_{1}^\gamma , \qquad \partial_j^\gamma (P^\gamma_{(2)}|_{\R^2_\text{reg}}) = (P^\gamma_{(2)} \partial_j)|_{\R^2_\text{reg}} \quad \text{for } j=1,2. \]
$\psi_{\lambda_1,\lambda_2} = P^\gamma_{(2)} \e^{\ii(\lambda_1,\lambda_2)}$ spans the solution space of the system $\partial_1^\gamma  f = \ii \lambda_1 f$, $\partial_2^\gamma  f = \ii \lambda_2 f$ in $\ca{CB}^\infty(\R^2_\text{reg})$.
We have
\begin{equation}
\psi_{\lambda_1,\lambda_2} = \e^{\ii(\lambda_1,\lambda_2)} + \gamma \theta_{21} I_{21}  \e^{\ii(\lambda_1,\lambda_2)} 
=  \e^{\ii(\lambda_1,\lambda_2)} + \frac{\ii \gamma \theta_{21}}{\lambda_1-\lambda_2}  \left( \e^{\ii(\lambda_1,\lambda_2)} -\e^{\ii(\lambda_2,\lambda_1)} \right). \label{eqn:psiexample}
\end{equation}

The symmetrized eigenfunction is now given by
\begin{align*}
\Psi_{\lambda_1,\lambda_2} = \frac{1}{2} \left( \psi_{\lambda_1,\lambda_2}  \! + \! s_1 \psi_{\lambda_1,\lambda_2} \right)  
%&=\frac{1}{2} \left( \theta_{12}\left( \frac{\lambda_1\! - \!\lambda_2\! - \!\ii \gamma}{\lambda_1\! - \!\lambda_2} \e^{\ii(\lambda_1,\lambda_2)} \! + \! \frac{\lambda_1\! - \!\lambda_2\! + \!\ii \gamma}{\lambda_1\! - \!\lambda_2}\e^{\ii(\lambda_2,\lambda_1)} \right) \! + \! \right. \\
%& \qquad \! + \! \left. \theta_{21}\left(\frac{\lambda_1\! - \!\lambda_2\! + \!\ii \gamma}{\lambda_1\! - \!\lambda_2} \e^{\ii(\lambda_1,\lambda_2)} \! + \! \frac{\lambda_1\! - \!\lambda_2\! - \!\ii \gamma}{\lambda_1\! - \!\lambda_2} \e^{\ii(\lambda_2,\lambda_1)} \right)  \right) \displaybreak[2] \\ &
= \frac{1}{2} \left(  \frac{\lambda_1\! - \!\lambda_2\! - \! \sgn_{1\,2} \ii \gamma}{\lambda_1\! - \!\lambda_2} \e^{\ii(\lambda_1,\lambda_2)} \! + \! \frac{\lambda_1\! - \!\lambda_2\! + \! \sgn_{1\,2} \ii \gamma}{\lambda_1\! - \!\lambda_2}\e^{\ii(\lambda_2,\lambda_1)}  \right),
\end{align*}
where $\sgn_{1\,2}=\theta_{12}-\theta_{21}$.
By restricting $\Psi_{\lambda_1,\lambda_2}$ to $J^2$, where $J=[x^+,x^-]$ with $x^--x^+=L$, and imposing the Bethe ansatz equations
\[ \frac{\lambda_1-\lambda_2+\ii \gamma}{\lambda_1-\lambda_2-\ii \gamma} = \e^{\ii \lambda_1 L} = \e^{-\ii \lambda_2 L}, \]
$\Psi_{\lambda_1,\lambda_2}|_{J^2}$ can be extended to a function $L$-periodic in each variable.
However such an extension for $\psi_{\lambda_1,\lambda_2}|_{J^2}$ does not exist, unless $\gamma =0$.
Indeed, from \eqref{eqn:psiexample} it follows that periodicity in the first argument, viz. $\psi_{\lambda_1,\lambda_2}(x^+,x)$ = $\psi_{\lambda_1,\lambda_2}(x^-,x)$ for $x \in J$, amounts to
%\[ \frac{\lambda_1-\lambda_2+\ii \gamma}{\lambda_1-\lambda_2} \e^{\ii \lambda_1 x^+} \e^{\ii \lambda_2 x} - \frac{\ii \gamma}{\lambda_1-\lambda_2} \e^{\ii \lambda_2 x^+} \e^{\ii \lambda_1 x} = \e^{\ii \lambda_1 x^-} \e^{\ii \lambda_2 x} \]
%i.e.
\[ \e^{\ii (\lambda_1 - \lambda_2) x} = \frac{\lambda_1-\lambda_2+\ii \gamma}{\ii \gamma} \e^{\ii (\lambda_1-\lambda_2) x^+} - \frac{\lambda_1-\lambda_2}{\ii \gamma} \e^{\ii (\lambda_1 x^--\lambda_2 x^+)}  \]
For this to hold for all $x \in J$, it is necessary that $\lambda_1=\lambda_2$, which leads to a contradiction as follows.
By l'H\^opital's rule we have
\[
\psi_{\lambda,\lambda}(x_1,x_2) := \lim_{\lambda_1, \lambda_2 \to \lambda} \psi_{\lambda_1,\lambda_2}(x_1,x_2) 
= \e^{\ii \lambda(x_1+x_2)} + \begin{cases} 0, & x_1>x_2, \\ \gamma (x_2-x_1) \e^{\ii \lambda(x_1+x_2)}, & x_2>x_1. \end{cases} \]
Hence $\psi_{\lambda,\lambda}(x^+,x) = \psi_{\lambda,\lambda}(x^-,x)$ for all $x \in J$ implies that $1+\gamma(x-x^+)=\e^{\ii \lambda L}$ for all $x \in J$, contradicting $\gamma \ne 0$. Periodicity in the second argument can be ruled out in the same way.\\

The non-symmetric creation operators can be used to construct $\psi_{\lambda_1,\lambda_2}$ from the pseudovacuum $\Vac =1$ by means of $\psi_{\lambda_1,\lambda_2} = b^-_{\lambda_2} b^-_{\lambda_1} \Vac = b^+_{\lambda_1} b^+_{\lambda_2} \Vac $ and are given by
\begin{description}
\item[(From $N=0$ to $N=1$)] For $f \in \C$ and $x \in \R$ we have $(b^\pm_\mu f)(x) = \e^{\ii \mu x} f$.
\item[(From $N=1$ to $N=2$)] For $f \in \ca{CB}^\infty(\R)$ and $(x_1,x_2) \in \R^2$ we have
\begin{align*}
(b^-_\mu f)(x_1,x_2) &= \e^{\ii \mu x_2} f(x_1) + \begin{cases} 0, & x_1>x_2, \\ \gamma \int_{x_1}^{x_2} \! \! \dd y \e^{\ii \mu(x_1+x_2-y)} f(y), & x_2>x_1, \end{cases}  \\
(b^+_\mu f)(x_1,x_2) &= \e^{\ii \mu x_1} f(x_2) + \begin{cases} 0, & x_1>x_2, \\ \gamma \int_{x_1}^{x_2} \! \! \dd y \e^{\ii \mu(x_1+x_2-y)} f(y), & x_2>x_1. \end{cases}
\end{align*}
\end{description}

The operators $a^\pm_\mu$ are given by
\begin{description}
\item[($N=0$)] For $f \in \C$, $\epsilon = \pm$ we have $a^\epsilon_\mu f = \e^{\ii \mu x^\epsilon} f$. 
\item[($N=1$)] For $f \in \f h_1$ and $x \in J$ we have
\begin{align*}
(a^-_\mu f)(x) &= \e^{\ii \mu x^-}  f(x) + \gamma \int_{x}^{x^-} \! \! \dd y \e^{\ii \mu (x^-+x-y)} f(y) ,  \\
(a^+_\mu f)(x) &= \e^{\ii \mu x^+}  f(x) + \gamma \int_{x^+}^{x} \! \! \dd y \e^{\ii \mu (x^++x-y)} f(y).
\end{align*}
\item[$(N=2)$] For $f \in \f h_2$ and $(x_1,x_2) \in J^2$ we have
\begin{align*}
(a^-_\mu f)(x_1,x_2)  &= \e^{\ii \mu x^-} f(x_1,x_2) + \\
& \qquad +\gamma \int_{x_1}^{x^-} \! \! \dd y \e^{\ii \mu(x^-+x_1-y)} f(y,x_2) + \gamma \int_{x_2}^{x^-} \! \! \dd y \e^{\ii \mu(x^-+x_2-y)} f(x_1,y) + \\
& \qquad + 
\begin{cases} \gamma^2 \int_{x_1}^{x^-} \! \! \dd y_1 \int_{x_2}^{x_1} \! \! \dd y_2 \e^{\ii \mu(x^-+x_1+x_2-y_1-y_2)} f(y_1,y_2), & x_1>x_2, \\ 
\gamma^2 \int_{x_1}^{x_2} \! \! \dd y_1 \int_{x_2}^{x^-} \! \! \dd y_2 \e^{\ii \mu(x^-+x_1+x_2-y_1-y_2)} f(y_1,y_2),  &  x_2>x_1,
\end{cases}  \\
(a^+_\mu f)(x_1,x_2) &= \e^{\ii \mu x^+} f(x_1,x_2) + \\
& \qquad + \gamma \int_{x^+}^{x_1} \! \! \dd y \e^{\ii \mu(x^++x_1-y)} f(y,x_2) +  \gamma \int_{x^+}^{x_2} \! \! \dd y \e^{\ii \mu(x^++x_2-y)} f(x_1,y) +  \\
& \qquad + \begin{cases} \gamma^2 \int_{x_2}^{x_1} \! \! \dd y_1 \int_{x^+}^{x_2} \! \! \dd y_2 \e^{\ii \mu(x^++x_1+x_2-y_1-y_2)} f(y_1,y_2), & x_1>x_2, \\
\gamma^2  \int_{x^+}^{x_1} \! \! \dd y_1 \int_{x_1}^{x_2} \! \! \dd y_2 \e^{\ii \mu(x^++x_1+x_2-y_1-y_2)} f(y_1,y_2) , & x_2>x_1,
\end{cases}
\end{align*}
\end{description}
leading to the following expressions for $a^\pm_\mu \psi_{\lambda_1,\lambda_2}$ as per \rft{thm:apsi}:
\begin{align*}
a^+_\mu \psi_{\lambda_1,\lambda_2}&= \frac{\mu\!-\!\lambda_1\!+\!\ii \gamma}{\mu\!-\!\lambda_1} \frac{\mu\!-\!\lambda_2\!+\!\ii \gamma}{\mu\!-\!\lambda_2} \e^{\ii \mu x^+} \psi_{\lambda_1,\lambda_2} + \frac{\!-\!\ii \gamma}{\mu\!-\!\lambda_1} \frac{\lambda_1\!-\!\lambda_2\!+\!\ii \gamma}{\lambda_1\!-\!\lambda_2} \e^{\ii \lambda_1 x^+}\psi_{\mu,\lambda_2} + \\
& \qquad + \frac{\mu\!-\!\lambda_1\!+\!\ii \gamma}{\mu\!-\!\lambda_1} \frac{\!-\!\ii \gamma}{\mu\!-\!\lambda_2} \e^{\ii \lambda_2 x^+} \psi_{\lambda_1,\mu}
+ \frac{\!-\!\ii \gamma}{\mu\!-\!\lambda_1} \frac{\!-\!\ii \gamma}{\lambda_1\!-\!\lambda_2} \e^{\ii \lambda_2 x^+}  \psi_{\mu,\lambda_1}, \displaybreak[2] \\
a^-_\mu \psi_{\lambda_1,\lambda_2} &=  \frac{\lambda_1\!-\!\mu\!+\!\ii \gamma}{\lambda_1\!-\!\mu} \frac{\lambda_2\!-\!\mu\!+\!\ii \gamma}{\lambda_2\!-\!\mu}  \e^{\ii \mu x^-} \psi_{\lambda_1,\lambda_2} +  \frac{-\ii \gamma}{\lambda_1 \! - \! \mu} \frac{\lambda_2\!-\!\mu\!+\!\ii \gamma}{\lambda_2\!-\!\mu}   \e^{\ii \lambda_1 x^-} \psi_{\mu,\lambda_2} + \\
& \qquad + \frac{\lambda_1 \! - \! \lambda_2 \! + \! \ii \gamma}{\lambda_1 \! - \! \lambda_2}  \frac{-\ii \gamma}{\lambda_2\!-\!\mu} \e^{\ii \lambda_2 x^-} \psi_{\lambda_1,\mu} +
\frac{-\ii \gamma}{\lambda_1 \! - \! \lambda_2} \frac{-\ii \gamma}{\lambda_2 \! - \! \mu} \e^{\ii \lambda_1 x^-}
\psi_{\lambda_2,\mu}.
\end{align*}

\section{Commutation relations between $\hat e^-_{\mu;\bm i}$ and $I_{j \, k}$} \label{sec:eI}

In this appendix we present lemmas used in the proof of the key result \rft{thm:bpropcommrel}.
As $\mu \in \C$ is arbitrary but fixed we will use the following shorthand: $\hat e^-_{\bm i}= \hat e^-_{\mu;\bm i}$ and $b^- = b^-_\mu$.

\begin{lem} \label{lem:eI1}
Let $n=1,\ldots,N$, $\bm i \in \f I^n_{[1,N]}$ and $j<i_1$.
Then 
\begin{align} 
I'_{N+1 \, j} \hat e^-_{\emptyset} &= \hat e^-_j; \label{eqn:eI10} \\
I'_{N+1 \, j} \hat e^-_{\bm i} &= \hat e^-_{j \, \bm i} + \hat e^-_{\bm i} I'_{i_1 \, j}.\label{eqn:eI11} 
\end{align}
\end{lem}

\begin{proof}
\eqref{eqn:eI10} is immediate; for \eqref{eqn:eI11} it is sufficient to prove that
\[ (I_{N+1 \, j} \hat e^-_{\bm i}f)(\bm x) = (\hat e^-_{j \, \bm i} f)(\bm x) + ( \hat e^-_{\bm i} \theta_{i_1 \, j} I_{i_1 \, j} f)(\bm x) \]
for $f \in \ca C(\R^N)$ and $\bm x = (x_1,\ldots,x_{N+1}) \in \R^{N\!+\!1}$ such that $x_{N+1}>x_{i_1}>\ldots>x_{i_n}$.
We have
\begin{align*}
(I_{N\!+\!1 \, j} \hat e^-_{\bm i} f)(\bm x) &=\int_{x_j}^{x_{N\!+\!1}}\!\! \dd y_0 \int_{x_{i_1}}^{x_{N\!+\!1}\!+\!x_j\!-\!y_0} \! \! \dd y_1 \tilde f_{y_0,y_1}(\bm x) ; \\
(\hat e^-_{j \, \bm i} f)(\bm x)&=\int_{x_j}^{x_{N\!+\!1}}\!\! \dd y_0 \int_{x_{i_1}}^{x_j} \! \! \dd y_1 \tilde f_{y_0,y_1}(\bm x),
\end{align*}
where 
\[ \tilde f_{y_0,y_1}(\bm x) = \int_{x_{i_2}}^{x_{i_1}} \! \! \dd y_2 \cdots \int_{x_{i_n}}^{x_{i_{n\!-\!1}}} \! \! \dd y_n \e^{\ii \mu \bigl(x_{N\!+\!1}+x_j-y_0+\sum_{m=1}^n (x_{i_m}-y_m) \bigr)}f(x_1,\ldots,\underset{(i_1)}{y_1}, \ldots, \underset{(i_n)}{y_n},\ldots,x_N).\] On the other hand,
\begin{align*}
(\hat e^-_{\bm i} \theta_{i_1 \, j} I_{i_1 \, j} f)(\bm x) 
&= \int_{x_{i_1}}^{x_{N+1}} \dd y_1 \cdots \int_{x_{i_n}}^{x_{i_{n-1}}} \dd y_n \e^{\ii \mu(x_{N+1} + \sum_{m=1}^n (x_{i_m} -y_m))} \cdot \\
& \qquad \cdot (\theta_{i_1 \, j} I_{i_1 \, j} f)(x_1,\ldots,\underset{(i_1)}{y_1}, \ldots, \underset{(i_n)}{y_n},\ldots,x_N) \\
&= \int_{x_j}^{x_{N+1}} \dd y_1 \int_{x_{i_2}}^{x_{i_1}} \dd y_2 \cdots \int_{x_{i_n}}^{x_{i_{n-1}}} \dd y_n \e^{\ii \mu \bigl(x_{N+1}+ \sum_{m=1}^n (x_{i_m}-y_m)\bigr)} \cdot \\
& \qquad \cdot \int_{x_j}^{y_1} \dd y_0 f(x_1,\ldots,\underset{(j)}{y_0},\ldots, \underset{(i_1)}{x_j+y_1-y_0}, \ldots, \underset{(i_n)}{y_n},\ldots,x_N)   \\
&= \int_{x_j}^{x_{N+1}}\! \! \dd y_0 \int_{y_0}^{x_{N+1}} \! \! \dd y_1 \int_{x_{i_2}}^{x_{i_1}} \! \! \dd y_2 \cdots \int_{x_{i_n}}^{x_{i_{n-1}}}  \! \! \dd y_n \e^{\ii \mu \bigl(x_{N+1}+ \sum_{m=1}^n (x_{i_m}-y_m)\bigr)} \cdot \\
& \qquad \cdot  f(x_1,\ldots,\underset{(j)}{y_0},\ldots, \underset{(i_1)}{x_j+y_1-y_0}, \ldots, \underset{(i_n)}{y_n},\ldots,x_N)  \\
&= \int_{x_j}^{x_{N+1}} \dd y_0 \int_{x_j}^{x_j+x_{N+1}-y_0} \dd y_1 \tilde f_{y_0,y_1}(\bm x),
\end{align*}
where we have changed the order of integration of the integrals over $y_0$ and $y_1$, and substituted $y_1 \to x_j+y_1-y_0$. 
The following decomposition (ignoring sets of measure zero) for $x_{N+1} > x_j > x_{i_1}$ completes the proof: 
\begin{align*} 
&\set{(y_0,y_1) \in \R^2}{x_{N+1} > y_0 > x_j, \, x_j+x_{N+1}-y_0 > y_1 > x_{i_1}}  =\\
&= \set{(y_0,y_1) \in \R^2}{x_{N+1} > y_0 > x_j > y_1 > x_{i_1}} \cup \\
& \qquad \cup \set{(y_0,y_1) \in \R^2}{x_{N+1} > y_0 > x_j, \, x_j+x_{N+1}-y_0 > y_1 > x_j}. \tag*{\qedhere}
 \end{align*}
\end{proof}

For $n=0,\ldots,N$ and $\bm i \in \f i^n_{[1,N]}$ introduce
for $j=1,\ldots,i_{n-1}$ such that $j \ne i_k$ for any $k$, the ``next label'' function $\n_j(\bm i) = \min\set{ i_m }{ i_m>j }$,
and
\[ I'_{\bm i} := \theta_{\bm i} I_{i_1 \, i_2} \cdots I_{i_1 \, i_n} \in \End(\ca C(\R^N)). \]
Note that $ I'_{\bm i} = I'_{i_1 \, i_2} I'_{i_2 \cdots i_n}$. 
Furthermore, for $p \leq q \leq N$ and $\bm k \in \f I^q_{[1,N]}$, introduce the following set of ``decompositions'' of $\bm k$:
\[ \dcm_p(\bm k) = \set{(\bm i , \bm j) \in \f I^{q-p}_{[1,N]} \times \f I^p_{[1,N]}}{\forall l \exists m: i_l = k_m, \, \forall l \exists m: j_l = k_m, \, \forall l \forall m:\,  i_l \ne j_m, i_{q-p} = k_q}, \]
i.e. the set of pairs of disjoint ordered tuples such that the second member has the prescribed length $p$, their union (seen as sets) is the given ordered tuple $\bm k$ and the last entry of the first member is the last entry of the given tuple.

\begin{lem} \label{lem:eI2}
Let $q=1,\ldots,N$ and $\bm k \in \f I^q_{[1,N]}$. Then
\[ I'_{N+1 \, \bm k} \hat e^-_\emptyset = \sum_{p=0}^{q-1} \sum_{(\bm i,\bm j) \in \dcm_p(\bm k)} \! \! \hat e^-_{\bm i} I'_{\n_{j_1}(\bm i) \, j_1} \cdots I'_{\n_{j_p}(\bm i) \, j_p}. \]
\end{lem}
\begin{proof}
By induction on $q$. 
The case $q=1$ is equivalent to \eqref{eqn:eI10}; the set $\dcm_{p}(\bm k)$ is empty unless $p=0$, in which case $\dcm_0((k_1)) = \{((i_1),())\}$.
With $\bm k'= (k_2,\ldots,k_q)$, the induction step yields
\[
I'_{N+1 \, \bm k} \hat e^-_\emptyset = I'_{N+1 \, k_1} I'_{N+1 \, \bm k'} \hat e^-_\emptyset =\sum_{p=0}^{q-2} \sum_{(\bm i,\bm j) \in \dcm_p(\bm k')} 
I'_{N\!+\!1 \, k_1}  \hat e^-_{\bm i} I'_{\n_{j_1}(\bm i) \, j_1} \cdots I'_{\n_{j_p}(\bm i) \, j_p}, \]
where we have applied the induction hypothesis.
Using \eqref{eqn:eI11} we obtain
\[ I'_{N+1 \, \bm k} \hat e^-_\emptyset = \sum_{p=0}^{q-2} \Bigl( \sum_{(\bm i,\bm j) \in \dcm_p(\bm k')}  \hat e^-_{k_1 \, \bm i} I'_{\n_{j_1}(\bm i) \, j_1} \cdots I'_{\n_{j_p}(\bm i) \, j_p} + \sum_{(\bm i,\bm j) \in \dcm_p(\bm k')}  \hat e^-_{\bm i} I'_{i_1 \, k_1} I'_{\n_{j_1}(\bm i) \, j_1} \cdots I'_{\n_{j_p}(\bm i) \, j_p} \Bigr).\]
One completes the induction step by re-writing the first summation as one over $((k_1, \bm i), \bm j)$ $\in$ $\dcm_p(\bm k)$ and the second summation as one over $(\bm i, (k_1,\bm j))$ $\in$ $\dcm_{p+1}(\bm k)$.
\end{proof}

\begin{lem} \label{lem:eI3}
Let $m=1,\ldots,N+1$. On $s_N \cdots s_m \R^{N\!+\!1}_+$ we have
\[ P^\gamma_{(N\!+\!1)} \hat e^-_\emptyset = \sum_{n \geq 0} \gamma^n \sum_{\bm i \in \f I^n_{[m,N]}} \hat e^-_{\bm i} \sum_{p \geq 0} \sum_\atop{\bm j \in \f I^p_{[m,i_n)}}{\forall r,s: \, i_r \ne j_s} \Bigl(\gamma I'_{\n_{j_1}(\bm i) \, j_1} \Bigr) \cdots \Bigl(\gamma I'_{\n_{j_p}(\bm i) \, j_p} \Bigr). \]
\end{lem}

\begin{proof}
The condition $\bm x \in s_N \cdots s_m \R^{N\!+\!1}_+$ is equivalent to $x_1>\ldots>x_{m\!-\!1}>x_{N+1}>x_m>\ldots>x_N$. 
From \eqref{eqn:proprestricted2} we obtain
\[ P^\gamma_{(N\!+\!1)} = \sum_{q \geq 0} \gamma^q \sum_{\bm k \in \f I^q_{[m,N]}} I_{N\!+\!1 \, k_1} \cdots I_{N\!+\!1 \, k_q},\]
so that by virtue of \rfl{lem:eI2} we find that
\begin{align*} 
P^\gamma_{(N\!+\!1)} \hat e^-_\emptyset &= \sum_{q \geq 0} \gamma^q \sum_{\bm k \in \f I^q_{[m,N]}} I'_{N\!+\!1 \, \bm k} \hat e^-_\emptyset = \sum_{q \geq 0} \gamma^q \sum_{\bm k \in \f I^q_{[m,N]}} \sum_{n=0}^q \sum_{(\bm i,\bm j) \in \dcm_p(\bm k)}  \! \! \hat e^-_{\bm i} I'_{\n_{j_1}(\bm i) \, j_1}  \cdots I'_{\n_{j_p}(\bm i) \, j_p} \\
&=\sum_{n \geq 0} \gamma^n  \sum_{p \geq 0}  \sum_{\bm k \in \f I^{n+p}_{[m,N]}}\sum_{(\bm i,\bm j) \in \dcm_p(\bm k)} \! \! \hat e^-_{\bm i} \Bigl( \gamma I'_{\n_{j_1}(\bm i) \, j_1} \Bigr) \cdots \Bigl( \gamma I'_{\n_{j_p}(\bm i) \, j_p} \Bigr),
\end{align*}
(we have substituted $n=q-p$) and we settle the lemma by combining the summation over $\bm k$ with the summation over $(\bm i,\bm j)$, noting that $i_n = k_{n+p}$ is equivalent to $i_n > j_p$.
\end{proof}

Given $n=0,\ldots,N$, fix $\bm i \in \f I^n_{[1,N]}$. 
We may split $\bm i$ into subtuples of \emph{consecutive runs}; e.g. the consecutive runs of the tuple $(1,3,4,5,7,8)$ are the tuples $(1)$, $(3,4,5)$ and $(7,8)$. 
We will denote the number of consecutive runs by $l(\bm i)$.
An entry $i_r$ is termed an \emph{initial entry} or a \emph{terminal entry} if $i_r-1$ or $i_r+1 \ne i_m$ for any $m$, respectively; in other words if $i_r$ is the first or last entry of a consecutive run of $\bm i$, respectively. 
We denote by $(\sigma_1(\bm i),\ldots,\sigma_{l(\bm i)}(\bm i))$ the ordered subtuple of $\bm i$ consisting of its initial entries, and by $(\tau_1(\bm i),\ldots,\tau_{l(\bm i)}(\bm i))$ the ordered subtuple of $\bm i$ consisting of its terminal entries. We have $\sigma_k(\bm i) \leq \tau_k(\bm i)$ for $k=1,\ldots,l(\bm i)$ and $\tau_k(\bm i)<\sigma_{k+1}(\bm i)-1$ for $k=1,\ldots,l(\bm i)-1$. 
Note that for $j \ne i_k$, $j < i_n$, $\n_j(\bm i)$ is an initial entry of $\bm i$.
If it is clear from the context what $\bm i$ is, we will not specify this in the notation.
For example, let $N=10$ and $\bm i = (2,4,5,6,9,10)$. Then $l=3$, $(\sigma_1,\sigma_2,\sigma_3)=(2,4,9)$ and $(\tau_1,\tau_2,\tau_3)=(2,6,10)$. For $j=1,3,7,8$ we have $\n_j(\bm i) = 2,4,9,9$.

\begin{lem} \label{lem:eI4}
Let $m=1,\ldots,N+1$, $n=0,\ldots,N$ and $\bm i \in \f I^n_{[m,N]}$. 
On $s_N \cdots s_m \R^{N\!+\!1}_+$ we have
\[ \sum_{p \geq 0} \sum_\atop{\bm j \in \f I^p_{[m,i_n)}}{\forall r,s: \, i_r \ne j_s} \Bigl(\gamma I'_{\n_{j_1} \, j_1} \Bigr) \cdots \Bigl(\gamma I'_{\n_{j_p} \, j_p} \Bigr) =  \prod_{k=1}^l \Bigl( 1+\gamma I'_{\sigma_k \, \tau_{k\!-\!1}\!+\!1}\Bigr) \cdots \Bigl( 1+\gamma I'_{\sigma_k \, \sigma_k\!-\!1}\Bigr),\]
where %$(\sigma_1,\ldots,\sigma_l)$ and $(\tau_1,\ldots,\tau_l)$ are the initial and terminal entries of $\bm i$, respectively, and 
$\tau_0=m-1$. 
In the case that $\sigma_1=m$, the $(k=1)$-factor in the product over $k$ is understood to equal $1$.
%, and all other factors involve at least $1+\gamma I'_{\sigma_k \, \tau_{k\!-\!1}\!+\!1}$.
Hence,
\[ P^\gamma_{(N\!+\!1)} \hat e^-_{\bm i} = \sum_{n \geq 0} \gamma^n \sum_{\bm i \in \f I^n_{[m,N]}}\hat e^-_{\bm i} \prod_{k=1}^{l(\bm i)} \Bigl( 1+\gamma I'_{\sigma_k\!(\bm i) \, \tau_{k\!-\!1}\!(\bm i)\!+\!1}\Bigr) \cdots \Bigl( 1+\gamma I'_{\sigma_k\!(\bm i) \, \sigma_k\!(\bm i)\!-\!1}\Bigr). \]
\end{lem}

\begin{proof}
The idea is to place the different entries of $\bm j$ in different ``bins'' determined by the $l$ consecutive runs in $\bm i$; this means in each bin the entry $\n_{j_p}$ is the same, namely $\sigma_k$. This yields the first formula. By virtue of \rfl{lem:eI3} the expression for $P^\gamma_{(N\!+\!1)} \hat e^-_\emptyset$ follows.
\end{proof}

\begin{lem} \label{lem:eI5}
Let $m=1,\ldots,N+1$, $n=0,\ldots,N$ and $\bm i \in \f I^n_{[m,N]}$. 
With the same notations as in \rfl{lem:eI4}, on $s_N \cdots s_m \R^{N\!+\!1}_+$ we have
\[ \hat e^-_{\bm i} = \hat e^-_{\bm i} \! \sum_{m_1 = \tau_0\!+\!1}^{\sigma_1} \! \! \! \cdots\!\! \! \sum_{m_l = \tau_{l\!-\!1}\!+\!1}^{\sigma_l} \! \! \theta_{1 \, \ldots \, m\!-\!1} \left( \prod_{k=1}^l \theta_{\tau_{k\!-\!1} \, \ldots \, m_k\!-\!1 \, \sigma_k \, m_k \, \ldots \, \sigma_k\!-\!1 \, \sigma_k\!+\!1 \, \ldots \, \tau_k} \right) \theta_{\tau_l \, \ldots \, N}. \]
\end{lem}

\begin{proof}
Due to the step operator $\theta_{N\!+\!1 \, \bm i}$ incorporated in $\hat e^-_{\bm i}$ the function acted upon by $\hat e^-_{\bm i}$ vanishes in certain alcoves. 
By definition, $\hat e^-_{\bm i}$ introduces $n$ integrations whose variables $y_1,\ldots,y_n$ replace the $x_{i_1} ,\ldots, x_{i_n}$ in the argument of the function acted upon, and are bounded by $x_{i_{l\!-\!1}}>y_l>x_{i_l}$ for $l=1,\ldots,n$. 
The intervals over which the $y_l$ run can be split up into intervals bounded by neighbouring $x_j$.
\end{proof}

\begin{lem} \label{lem:eI6}
Let $m=1,\ldots,N+1$, $n=0,\ldots,N$ and $\bm i \in \f I^n_{[m,N]}$. 
With the same notations as in \rfl{lem:eI4}, on $s_N \cdots s_m \R^{N\!+\!1}_+$ we have
\begin{align*} 
\hat e^-_{\bm i}  P^\gamma_{(N)} &= \hat e^-_{\bm i} \sum_{m_1 = \tau_0\!+\!1}^{\sigma_1} \cdots \sum_{m_l = \tau_{l\!-\!1}\!+\!1}^{\sigma_l} \prod_{k=1}^l \theta_{\tau_{k\!-\!1} \, \ldots \, m_k\!-\!1 \, \sigma_k \, m_k \, \ldots \, \sigma_{k\!-\!1} } (1+\gamma I_{\sigma_k \, m_k}) \cdots (1+\gamma I_{\sigma_k \, \sigma_k\!-\!1}) .
\end{align*}
\end{lem}

\begin{proof}
We remark that the product of $\theta$-operators in \rfl{lem:eI5} defines precisely one alcove, given by the inequality
\begin{align*} 
x_1 \!>\! \ldots \!>\! x_{m\!-\!1} & \!>\! \ldots \!>\! x_{m_1\!-\!1} \!>\! x_{\sigma_1} \!>\! x_{m_1} \!>\! \ldots \!>\! \widehat{x_{\sigma_1}} \!>\! \ldots \!>\! x_{\tau_1} \!>\! \\
& \!>\! \ldots \!>\! x_{m_2\!-\!1} \!>\! x_{\sigma_2} \!>\! x_{m_2} \!>\! \ldots \!>\! \widehat{x_{\sigma_2}} \!>\! \ldots \!>\! x_{\tau_2} \!>\! \\
& \qquad \vdots\\
& \!>\! \ldots \!>\! x_{m_l\!-\!1} \!>\! x_{\sigma_l} \!>\! x_{m_l} \!>\! \ldots \!>\! \widehat{x_{\sigma_l}} \!>\! \ldots \!>\! x_{\tau_l}\!>\! \ldots \!>\! x_N, \end{align*}
i.e. the alcove
\[ \R^N_{m_1 \, \cdots \, m_l}(\bm i) := \left( \prod_{k=1}^l s_{\sigma_k \! - \! 1} \cdots s_{m_k} \right) \R^N_+; \]
Note that since $\tau_{k-1} < m_k \leq \sigma_k$, for different values of $k$ the cycles $s_{\sigma_k \! - \! 1} \cdots s_{m_k}$ commute, so the order of the product over $k$ is immaterial.
Using this commutativity, on $\R^N_{m_1 \, \cdots \, m_l}(\bm i)$ we have
\[ P^\gamma_{(N)} =  \left( \prod_{k=1}^l s_{\sigma_k \! - \! 1} \cdots s_{m_k} \right) \left( \prod_{k=1}^l s_{\sigma_k \! - \! 1}^\gamma \cdots s_{m_k}^\gamma \right)^{-1}
= \prod_{k=1}^l s_{\sigma_k \! - \! 1} \cdots s_{m_k} s_{m_k}^\gamma  \cdots s_{\sigma_k \! - \! 1}^\gamma .\]
Using \eqref{eqn:proprestricted2} we obtain the lemma.
\end{proof}

\begin{lem} \label{lem:eI7}
Let $m=1,\ldots,N+1$, $n=0,\ldots,N$ and $\bm i \in \f I^n_{[m,N]}$. 
With the same notations as in \rfl{lem:eI4}, on $s_N \cdots s_m \R^{N\!+\!1}_+$ we have
\begin{align*} 
\hat e^-_{\bm i}  P^\gamma_{(N)} &= \hat e^-_{\bm i} \prod_{k=1}^l (1+\gamma I'_{\sigma_k \, \tau_{k\!-\!1}\!+\!1}) \cdots (1+\gamma I'_{\sigma_k \, \sigma_k\!-\!1}) 
\end{align*}
and hence
\begin{align*} 
b^- P^\gamma_{(N)} &= \sum_{n \geq 0} \gamma^n \sum_{\bm i \in \f i^n_{[m,N]}} \hat e^-_{\bm i} \prod_{k=1}^{l(\bm i)} (1+\gamma I'_{\sigma_k\!(\bm i) \, \tau_{k\!-\!1}\!(\bm i)\!+\!1}) \cdots (1+\gamma I'_{\sigma_k\!(\bm i) \, \sigma_k\!(\bm i)\!-\!1}) .
\end{align*}
\end{lem}

\begin{proof}
By combining the step operators in \rfl{lem:eI6} for different values of $m_1,\ldots,m_l$.
\end{proof}

\bibliographystyle{hplain}
\bibliography{bibliography}

\end{document}